%% file: simplex-soup.tex
\documentclass{amsart}

\usepackage{ma}

\graphicspath{{./fig/}}

\newtheorem{theorem}{Theorem}
\newtheorem{definition}[theorem]{\sc{Definition}}
\newtheorem{remark}[theorem]{\sc{Remark}}
\newtheorem{lemma}[theorem]{\sc{Lemma}}
\newtheorem{proposition}[theorem]{\sc{Proposition}}
\usepackage{algorithm2e}

\title[Optimal transport between a simplex soup and a point cloud]
      {An algorithm for optimal transport between a simplex soup and a point
  cloud}

\author{Quentin M\'erigot}
\address{Laboratoire de math\'ematiques d'Orsay, Universit\'e Paris-Sud,
  Orsay, France}
\author{Jocelyn Meyron}
\address{GIPSA-Lab, Grenoble INP,
  Grenoble, France \& Laboratoire Jean Kuntzmann, Universit\'e Grenoble-Alpes,
  Grenoble, France}
\author{Boris Thibert}
\address{Laboratoire Jean Kuntzmann, Universit\'e Grenoble-Alpes,
  Grenoble, France}

\begin{document}

\begin{abstract}

We propose a numerical method to find the optimal transport map
between a measure supported on a lower-dimensional subset of $\Rsp^d$
and a finitely supported measure. More precisely, the source measure
is assumed to be supported on a simplex soup, i.e. on a union of
simplices of arbitrary dimension between $2$ and $d$. As in
[Aurenhammer, Hoffman, Aronov, Algorithmica \textbf{20} (1), 1998,
  61--76] we recast this optimal transport problem as the resolution
of a non-linear system where one wants to prescribe the quantity of
mass in each cell of the so-called \emph{Laguerre diagram}.  We prove
the convergence with linear speed of a damped Newton's algorithm to
solve this non-linear system. The convergence relies on two
conditions: (i) a genericity condition on the point cloud with respect
to the simplex soup and (ii) a (strong) connectedness condition on the
support of the source measure defined on the simplex soup.  Finally,
we apply our algorithm in $\Rsp^3$ to compute optimal transport plans
between a measure supported on a triangulation and a discrete measure.
We also detail some applications such as optimal quantization of a
probability density over a surface, remeshing or rigid point set
registration on a mesh.

%
%
\end{abstract}

\maketitle


\input{inc/1-introduction.tex}
\subsection*{Acknowledgements}
This work has been partially supported by the LabEx PERSYVAL-Lab
(ANR-11-LABX-0025-01) funded by the French program Investissement d’avenir and
by ANR-16-CE40-0014 - MAGA - Monge Amp\`ere et G\'eom\'etrie Algorithmique.
\input{inc/2-regularity.tex}

\input{inc/3-concavity.tex}
\input{inc/5-convergence.tex}
\input{inc/4-numerical-results.tex}

\bibliographystyle{amsplain}
\bibliography{algoForManifold}

\end{document}

%% file: inc/1-introduction.tex
\section{Introduction}

 In the last few years, optimal transport has received a lot of
 attention in mathematics (see e.g. \cite{villani2009optimal} and
 references therein), but also in computational geometry and in
 geometry processing because of the intimate connection between
 optimal transport maps for the quadratic cost and Power diagrams
 \cite{oliker1989numerical,aurenhammer1998minkowski,
   merigot2011multiscale, de2012blue,
   de2014intersection,levy2015numerical}.  By now, there exist
 efficient algorithms for computing the optimal transport between a
 piecewise-affine probability density on $\Rsp^d$ onto a finitely
 supported probability measure, a situation often referred to as
 semi-discrete optimal transport. In this article we look at a more
 singular setting where the source measure is not a probability
 density anymore, but is instead supported on a simplex soup, i.e.  a
 finite union of simplices in $\Rsp^d$. In the theoretical part of
 this article, we will allow the dimension of the simplices to range
 from $2$ to $d$. We call such a measure a \emph{simplicial
   measure}. The situation where one or more simplices in the soup
 have dimension strictly less than $d$ is difficult both in theory (as
 Brenier's theorem does not apply, and the optimal transport might not
 exist or not be unique) and in practice (the simplex could be
 included in the boundary of a Power cell, making the problem
 ill-posed).  Here, we propose a converging algorithm to solve the
 optimal transport problem in this degenerate setting.

\subsection{Optimal transport problem and Monge-Amp\`ere equation.} 
We first describe the general optimal transport problem between a
probability measure $\mu$ on $\Rsp^d$ and a probability measure $\nu$
supported on a point cloud of $\Rsp^d$. We always consider the
quadratic cost $c(x,y) = \nr{x-y}^2$.  The optimal transport problem
between $\mu$ and $\nu$ consists in finding a map $T: \Rsp^d\to Y$
that minimizes 
$
\int_{\Rsp^d} || x - T(x) ||^2
    d\mu(x) 
    $ 
    under the constraint that $ T_{\#} \mu =\nu$, where
    $ T_{\#} \mu $ denotes the pushforward of $ \mu $ by the map $ T
    $.
    When the target measure is finitely supported, i.e. $\nu =
    \sum_{1\leq i\leq N} \nu_i\delta_{y_i}$, this problem can be
    recast as a finite-dimensional non-linear system of equations
    involving the so-called Laguerre cells (see below)
    \cite{aurenhammer1998minkowski, gangbo1996geometry}. This idea can
    be traced back to Alexandrov and Pogorelov in
    convex geometry.

    More precisely, one can show that the optimal map $T : K \to Y$
    between $\mu$ and $\nu$ is of the form $T_\psi:x\mapsto
    \mathrm{argmin}_i \|x-y_i\|^2 + \psi_i$, where $(\psi_i)_{1\leq i
      \leq N}$ is a family of weights on $Y$. This implies that
    solving the optimal transport problem is equivalent to finding
    $\psi\in\Rsp^N$ such that $T_{\psi\#} \mu = \nu$. This last
    condition is equivalent to $G_i(\psi) = \nu_i$ for all $1\leq
    i\leq N$, where $G_i(\psi) := \mu(\Lag_i(\psi))$. Setting $G =
    (G_1,\hdots,G_N)$, the optimal transport problem between $\mu$ and
    $\nu=\sum_i \nu_i\delta_{y_i}$ amounts to the resolution of the
    finite-dimensional non-linear system of equations:
\begin{equation} \label{eq:MA}
    \tag{DMA}
  \hbox{ Find } \psi\in \Rsp^N \hbox{ such that }~~ G(\psi_1,\hdots,\psi_N) = (\nu_1,\hdots,\nu_N).
\end{equation}

\begin{remark}
When $\mu$ and $\nu$ are two probability densities on $\Rsp^d$,
Brenier's theorem asserts that $T = \nabla F$ is the gradient of a
convex function $F$. This function solves (in a suitable weak sense)
the non-linear differential equation $\nu(\nabla F(x)) \det(\D^2 F(x))
= \mu(x),$ which is called the \emph{Monge-Amp\`ere}
equation. Equation \eqref{eq:MA} can be regarded as a discretization
of this equation, hence the abbreviation.
\end{remark}
\begin{remark}
The non-linear system \eqref{eq:MA} admits a variational formulation,
which can be obtained as a consequence of Kantorovich's duality
theory, implying that $G$ is the gradient of a concave function,
e.g. \cite{aurenhammer1998minkowski,kitagawa2016newton}. We will not
use this fact here and don't develop this idea further. 
\end{remark}

From now on, we assume that the source measure $\mu$ is a simplicial
probability measure, as defined below.

\begin{definition}[Simplex soup]
\label{def:simplex-soup}
 A \emph{simplex soup} is a finite family $\Sigma$ of simplices of
 $\Rsp^d$. The dimension of
 a simplex $\sigma$ is denoted $d_\sigma$. The support of the
 simplex soup $\Sigma$ is the set $K = \cup_{\sigma\in \Sigma}
 \sigma$.
\end{definition}
\begin{definition}[Simplicial measure]
\label{def:simplicial-measure}
 We call \emph{simplicial measure} a measure $\mu = \sum_{\sigma \in
   \Sigma} \mu_\sigma$, where $\Sigma$ is a simplex soup, and where
 the measure $\mu_\sigma$ has density $\rho_\sigma$ with respect to
 the $d_\sigma$-dimensional Hausdorff measure on $\sigma$, i.e.  
$$
 \forall B\subseteq \Rsp^d\hbox{ Borel, } \mu(B)=\sum_{\sigma \in \Sigma} \int_{B \cap \sigma} \rho_\sigma(x) d\Haus^{d_\sigma}(x).
$$
\end{definition}

\subsection{Damped Newton's algorithm for semi-discrete optimal transport} 
We will solve the non-linear system \eqref{eq:MA} using the same
damped Newton's algorithm as in
\cite{mirebeau2015discretization,kitagawa2016newton}, which is
summarized in Algorithm \ref{algo:newton}. In this algorithm, we
denote by $A^{+}$ the \emph{pseudo-inverse} of the matrix $A$.  The
goal of this paper is to find conditions ensuring the convergence of
this algorithm in a finite number of steps. As usual for Newton's
methods, the convergence will be a natural consequence of the
$\Class^1$ regularity of $G$ and of a strict monotonicity property for
$D G$ (see Theorem~\ref{th:main} below). The strict monotonicity of
$G$ only holds near points $\psi\in \Rsp^N$ such that every Laguerre
cell contains a positive fraction of the mass, i.e. $\psi\in \K^+$
where
\begin{equation}
  \K^{+} = \{ \psi \in \Rsp^N \mid \forall i\in\{1,\hdots,N\},~ G_i(\psi) > 0\}. \label{eq:Kplus}
\end{equation}
The role of the damping step in Algorithm~1 (i.e. the choice of $\ell$
in the loop) is to ensure that $\psi^k$ always remain in $\K^+$. Also,
since $G$ is invariant under the addition of a constant to all weights, we
cannot expect \emph{strict} monotonicity of $G$ in all directions. We
 denote $\cstperp$ the orthogonal complement of the space of
constant functions on $ Y $ for the canonical scalar product on
$\Rsp^N$, i.e. $\cstperp = \{ v \in \Rsp^N \mid \sum_{1\leq i\leq N} v_i = 0 \}.$
Before summarizing the main properties of
$G$, we need an additional definition.

\begin{algorithm}[t]
\begin{description}
  \item[Input] A simplicial measure $\mu$, a finitely supported measure $\nu = \sum_{1\leq i\leq N} \nu_i\delta_{y_i}$,
    $\eta > 0$\\
\hspace{.37cm} A family of weights $\psi^0 \in\Rsp^N$ such that
      $\eps_0 := \min\left[\min_{i} G_i(\psi^0),~ \min_{i} \nu_i\right] > 0$
\item[While] $\nr{G(\psi^k) - \nu}  \geq \eta$
  \begin{itemize}
\item Compute $v^{k} = - \D G(\psi^k)^{+} (G(\psi^k) - \nu)$
\item Determine the minimum $\ell \in \Nsp$ such that $\psi^{k,\ell} :=
  \psi^k + 2^{-\ell} v^k$ satisfies
\begin{equation*}
\left\{
\begin{aligned}
&\min_{i} G_i(\psi^{k,\ell}) \geq \eps_0 \\
&\nr{G(\psi^{k,\ell}) - \nu} \leq (1-2^{-(\ell+1)}) \nr{G(\psi^k) - \nu}
\end{aligned}
\right.
\end{equation*}
\item Set $\psi^{k+1} = \psi^k + 2^{-\ell} v^k$ and $k\gets k+1$.
  \end{itemize}
\item[Output] A family of weights $\psi^k$ solving (\ref{eq:MA}) up to $\eta$, i.e. $\nr{G(\psi^k) - \nu} \leq \eta$.
\end{description}
\caption{Damped Newton's algorithm}
\label{algo:newton}
\end{algorithm}

\begin{definition}[Regular simplicial measure]
\label{def:reg-simplicial-measure}
A simplicial measure $\mu$ over $\bigcup_{\sigma\in\Sigma}\sigma$ is called \emph{regular} if
\begin{itemize}
\item the dimension of every simplex $ \sigma $ is $\geq 2.$
  \item for every $\sigma\in \Sigma$, $\rho_\sigma:\sigma\to\Rsp$ is
    continuous and $\min_{\sigma} \rho_\sigma > 0$.
  \item it is not possible to disconnect the support $K =
    \bigcup_{\sigma\in\Sigma} \sigma$ by removing a finite number of
    points, i.e.
    $\forall S\subseteq K \hbox{ finite, } K\setminus S \hbox{ is connected.}$
\end{itemize}
\end{definition}

\begin{theorem} \label{th:main}
    Assume $\mu$ is a regular simplicial measure and that the points $y_1,
    \ldots, y_n$ are in generic positions
    (according to Def.~\ref{def:Generic}). Then,
    \begin{itemize}
      \item  $G$ has class $\Class^1$
        on $\Rsp^N$.
        \item $G$ is strictly monotone in the following sense
          $$ \forall \psi\in\K^+, \forall v \in \cstperp\setminus\{0\},~~ \sca{\DD G(\psi) v}{v} < 0. $$
    \end{itemize}
\end{theorem}

The statement of this theorem is similar to Theorems~1.3 and 1.4 in
\cite{kitagawa2016newton}. However, the results of
\cite{kitagawa2016newton} were established under the assumption that
the Laguerre cells induced by the cost function are convex in some
``$c$-exponential chart'', which is the discrete version of the
so-called Ma-Trudinger-Wang property
\cite{ma2005regularity,loeper2009regularity}. In the setting
considered here, the Laguerre cells can be disconnected, so that we
cannot expect them to be convex in any chart. Consequently, the
strategy used in \cite{kitagawa2016newton} cannot be applied here, and
we need to find an alternative way to establish the regularity of
$G$. What we show here is that a mild genericity assumption on the
points $y_1,\hdots,y_N$ ensures that $G$ is $\Class^1$ even when the
source measure is singular, i.e. supported over a lower-dimensional
subset of $\Rsp^d$. The price to pay for this, however, is that we do
not (and cannot expect to) get quantitative estimates on the speed of
convergence of the algorithm as in \cite{kitagawa2016newton}. In
particular, the existence of $\tau^*$ in the following theorem is
obtained through a compactness argument.

\begin{theorem} \label{coro:convergence-newton}
    Under the hypotheses of the previous theorem, the proposed Damped
    Newton's algorithm converges in a finite number of steps.
    Moreover, the iterates of Algorithm~\ref{algo:newton} satisfy
         $$ \nr{G(\psi^{k+1}) - \nu} \leq \left( 1 - \frac{\tau^\star}{2} \right) \nr{G(\psi^{k}) - \nu}, $$
    where $ \tau^* \in ]0,1] $ depends on $\mu,\nu$ and $\epsilon_0$.
\end{theorem}

As we will see in Section 5, the behaviour of
Algorithm~\ref{algo:newton} seems better in practice: the number of
Newton's iterations is small even for large point sets. In our
numerical examples, the number of iterations never exceeds $16$.

\subsubsection{Related work.} The problem of optimal transport
between a probability density on $\Rsp^d$ and a finitely supported
measure has been considered in many works, and can be traced back to
Alexandrov and Pogorelov. The authors of \cite{oliker1989numerical}
proposed and analysed a coordinatewise-increment algorithm for a
problem similar but not quite equivalent to optimal transport --
namely, a Monge-Ampère equation with Dirichlet boundary
conditions. This coordinatewise-increment approach was extended to an
optimal transport setting in \cite{caffarelli1999problem}, leading to
a $\mathrm{O}(N^3/\eta)$ algorithm where $N$ is the number of Dirac
masses and $\eta$ is the desired error. Aurenhammer, Hoffmann and
Aronov \cite{aurenhammer1998minkowski} proposed a variational
formulation for semi-discrete optimal transport, but do not analyse its
algorithmic consequences further. This variational formulation was
combined with quasi-Newton \cite{merigot2011multiscale,levy2015numerical} or Newton's
\cite{de2012blue,su2013area} methods with good experimental results
but without convergence analysis. The convergence of a \emph{damped}
  Newton's algorithm was established first in
\cite{mirebeau2015discretization} for the Monge-Ampère equation with
Dirichlet condition and was extended to optimal transport for cost
functions satisfying the so-called Ma-Trudinger-Wang condition in
\cite{kitagawa2016newton}. None of these works deal with the singular
setting that we consider here, where the source measure might be
supported on a lower-dimensional subset of $\Rsp^d$. In particular, we
underline that in order to deal with surfaces embedded in $\Rsp^3$,
the authors of \cite{su2013area} first map them conformally in the
plane $\Rsp^2$.

\subsubsection{Applications.} We can apply our result to different settings where
the source and target measures are concentrated on lower-dimensional objects. 
We investigate at the end of this article applications such as \emph{optimal quantization} of a probability density over a surface, \emph{remeshing} or \emph{point set registration} on a
mesh. Another interesting application that we do not develop here is the optimal transport problem between
measures concentrated on graphs of functions~\cite{ma2005regularity}, which are lower-dimensional subsets of $ \Rsp^d $. Such problems occur for instance in signal analysis and machine
learning~\cite{thorpe2016transportation}. The cost involved in this setting is
of the form $ c(x,y)=|| x-y ||^2 + |f(x)-g(y)|^2 $. When the functions $f$ and
$g$ are strictly convex and their gradients are less than one,  the cost $c$
satisfies the Ma-Trudinger-Wang condition \cite{ma2005regularity} and we can
apply the results of \cite{kitagawa2016newton}. When $f$ and $g$ do not satisfy
these assumptions, our result shows that the damped Newton's algorithm still converges. 
%

\subsubsection{Outline.}
In Section \ref{sec:transport-plan}, we show the relation between solutions of
\eqref{eq:MA} and optimal transport plans. In Section \ref{sec:regularity},
we establish the regularity of the function $ G
$. Section \ref{sec:motonicity} is devoted to the proof of the strict motonicity
of $ G $. In Section \ref{sec:convergence}, we combine
the intermediate results to show the convergence of the damped Newton's
algorithm (Theorem \ref{coro:convergence-newton}). In Section
\ref{sec:numerical-results}, we present numerical
illustrations and applications of this algorithm.


%% file: inc/2-regularity.tex
\section{Optimal transport problem} \label{sec:transport-plan}
In this section, we show that the optimal transport problem considered
in this paper amounts to solving the system \eqref{eq:MA}. The results
mentioned here are very classical when the source measure is supported
on a full dimensional subset of $\Rsp^d$. Here, in order to handle
lower-dimensional simplex soups, we need to introduce a notion of
genericity. In the following, we denote by $ [x_0, \ldots, x_k] $ the
convex hull of the points $x_0, \ldots, x_k$.

\begin{definition}[Generic point set]\label{def:Generic}
A point set $\{ y_1, \ldots, y_N \} \subset \Rsp^d$ is in
\emph{generic position} with respect to a $k$-dimensional simplex
$\sigma=[x_0,\hdots,x_{k}]$ if the following condition holds
for every integer $p\in\{1,\hdots,\jocelyn{k}\}$, every $\ell
  \in\{1,\hdots,\min(d,N-1)\}$, every distinct $i_0,\hdots,i_\ell
  \in \{1,\hdots,N\}$ and every distinct
  $j_0,\hdots,j_p\in\{\jocelyn{0},\hdots,k\}$:
  \begin{equation}\label{eq:Transverse}
    \dim(\jocelyn{\{y_{i_1} - y_{i_0},\hdots,y_{i_{\ell}} - y_{i_0}\}}^{\perp} \cap \vect(x_{j_1} - x_{j_0}, \hdots, x_{j_p} - x_{j_0}))
    = \max(p-\ell, 0)
  \end{equation}
The point set is in generic position with respect to a simplex soup $K
= \cup_{\sigma\in \Sigma}\sigma$ if it is in generic position with
respect to all the simplices $\sigma\in\Sigma$.
\end{definition}
\begin{definition}[Power diagram]
  The $i$th power cell induced by weights $\psi\in \Rsp^N$ on
  a point set $\{y_1,\hdots,y_{N}\}$ is defined by
  $$\Pow_i(\psi):=\{ x \in \Rsp^d \mid \forall j \in \{1, \ldots, n\}, \nr{x - y_i}^2 + \psi_i \leq \nr{x - y_j}^2 + \psi_j \}.$$
\end{definition}

\begin{remark} Note that Laguerre cells are intersections of Power cells with the simplex soup, namely 
\begin{equation} \label{eq:lag}
\Lag_i(\psi) = \Pow_i(\psi) \cap K.
\end{equation}
  Condition \eqref{eq:Transverse} ensures in particular that for any
  choice of weights $(\psi_i)_{1\leq i\leq N}$ the
  $(d-\ell)$-dimensional facets of the Power diagram induced by
  $(y_i)_{1\leq i\leq N}, (\psi_i)_{1\leq i\leq N}$ intersect the
  $p$-dimensional facets of $\sigma$ in a trivial way, when
  $(d-\ell)+p\leq d$.
\end{remark}

We also need the following technical lemma that states that, under genericity, the Laguerre cells form a partition of a simplex soup almost everywhere. 

\begin{lemma}\label{lemma:lagij}
  Assume that $\mu$ is a simplicial measure and that $y_{1},\hdots,
  y_N$ is in generic position (Def~\ref{def:Generic}). Let
  $\psi\in\Rsp^N$ and define $\Lag_{i,j}(\psi) = \Lag_i(\psi)\cap\Lag_j(\psi)$ Then, 
  $$ \forall i\neq j,~~\mu(\Lag_{i,j}(\psi))
  = 0 \hbox{ and } \forall i,~~ \mu(\partial \Lag_i(\psi)) = 0.$$
\end{lemma}

\begin{proof} Let $\sigma=[x_0,\hdots,x_k]$ be a $k$-dimensional simplex in the support of $\mu$.
  Then, from the genericity assumption, one has $\dim(\vect(x_1 -
  x_0,\hdots,x_k-x_0) \cap \{y_i - y_j\}^\perp) = k - 1,$ so that in
  particular $\dim(\sigma \cap \Lag_{i,j}(\psi)) \leq k-1$. This gives
  $$ \mu_\sigma(\Lag_{i,j}(\psi)) = \int_{\sigma \cap \Lag_{i,j}(\psi)}
  \rho_\sigma(x) \dd \Haus^k(x) \dd x = 0. $$ Summing these equalities
  over $\sigma \in \Sigma$, we get $\mu(\Lag_{i,j}(\psi)) = 0$. The second equality then follows from
  $\partial \Lag_i(\psi) \subseteq \bigcup_{j\neq i} \Lag_{i,j}(\psi)$.
\end{proof}

\begin{definition}[Transport map]
  Let $\mu,\nu$ be two probability measures on $\Rsp^d$, and assume
  that $\nu$ is supported over a finite set $Y=\{y_1,\hdots,y_N\}$,
  i.e. $\nu = \sum_{1\leq i\leq N} \nu_i \delta_{y_i}$. A map $T: K
  \to Y$ is called a \emph{transport map} between $\mu$ and $\nu$ if
  $$ \forall i\in\{1,\hdots,N\},~ \mu(T^{-1}(y_i)) = \nu_i. $$
\end{definition}

The relation between solutions of \eqref{eq:MA} and optimal transport maps
is explained in the following proposition.

\begin{proposition} \label{prop:optimality-laguerre} Let  $\mu$ be a simplicial measure supported on $K$, and let  $y_{1},\hdots,
  y_N$ be in generic position (Def~\ref{def:Generic}). If $\psi \in
  \Rsp^N$ satisfies \eqref{eq:MA}, then, the map
  $$ T_\psi: x \in K\mapsto \arg\min_{i} \nr{x - y_i}^2 +\psi_i.$$ is
  well-defined $\mu$-a.e. and is an optimal transport
  map between $\mu$ and $\nu$.
  \label{prop:OT}
\end{proposition}
\begin{proof}
  The fact that $T_\psi$ is well-defined almost everywhere follows
  from Lemma~\ref{lemma:lagij}. Denote $\psi(y_i) := \psi_i$. Then, by
  definition of $T_\psi$, one has $\nr{x - T_\psi(x)}^2 +
  \psi(T_\psi(x)) \leq \nr{x - T(x)}^2 + \psi(T(x))$. Integrating this
  inequality gives
  $$ \int_{K} \nr{x - T_\psi(x)}^2 + \psi(T_\psi(x)) \dd \mu(x) \leq
  \int_{K} \nr{x - T(x)}^2 + \psi(T(x)) \dd \mu(x).$$ Since
  $T$ and $T_\psi$ are both transport maps between $\mu$ and $\nu$, a
  change of variable gives
  $$ \int_{K} \psi(\jocelyn{T_\psi(x}))  \dd \mu(x) = \sum_{1\leq i\leq N} \psi_i \nu_i = \int_{K} \psi(T(x))  \dd \mu(x).$$
  Subtracting this equality from the inequality above directly gives  the result.
\end{proof}

The goal of this article is to show the convergence of an algorithm able to
eficiently solve the system \eqref{eq:MA}. This relies on the regularity and a notion of strict monotonicity of the function $G$ that are studied in the following sections. 
%

\section{$\Class^1$ regularity of $G$}\label{sec:regularity}
The main result of this section is the following theorem that states that under genericity conditions, the function
$G:\Rsp^N\to\Rsp^N$ appearing in \eqref{eq:MA} is of class $\Class^1$.

\begin{theorem}\label{thm:regularity}
Let $\mu$ be a regular simplicial measure supported on a simplex soup
$\Sigma$ (as in Definition~\ref{def:reg-simplicial-measure}) and let $Y =
\{ y_1, \ldots, y_N \}$ be a generic point set. Then,
\begin{itemize}
  \item the function
$G$ appearing in \eqref{eq:MA} has class $\Class^1$ on $\Rsp^N$; 
\item denoting $\Lag_{i,j}(\psi) := \Lag_{i}(\psi) \cap \Lag_{j}(\psi)$, the derivatives of $G$ are given by
\begin{equation}
\begin{cases} \label{eq:HessPhi}
\frac{\partial G_i}{\partial \psi_j}(\psi) = 
\frac{1}{2 \nr{\jocelyn{\Pi_{\sigma^0}(y_i - y_j)}}} 
\sum_{\sigma\in \Sigma} \int_{\Lag_{i, j}(\psi) \cap \sigma} \rho_\sigma(x) \dd\Haus^{d_\sigma-1}(x) 
&\forall i\neq j\\
\frac{\partial G_i}{\partial \psi_i}(\psi) =
-\sum_{j\neq i} \frac{\partial^2 \jocelyn{G_i}}{\partial \jocelyn{\psi_j}}(\psi) &\forall i.
\end{cases}
\end{equation}
\jocelyn{where $ \Pi_{\sigma^0} : \Rsp^d \to \sigma^0 $ denotes the orthogonal
    projection on the linear subspace $ \sigma^0 $ tangent to $ \sigma $.}
\end{itemize}
\end{theorem}

\begin{remark}  Note that in contrast with Theorem~{4.1} in
  \cite{kitagawa2016newton}, the map $G$ is continuous on the whole
  space $\Rsp^N$ and not only on the set $\K^+$ defined in
  \eqref{eq:Kplus}. Without the genericity hypothesis, one cannot hope
  a global regularity result of this kind.
  \begin{itemize} \item Let $\mu$ be the uniform probability measure on $K = [0,1]^2\subseteq \Rsp^2$
    (union of two triangles), and let $y_1 = (\frac{1}{2},0)$, $y_2 = (-\frac{1}{2},0)$
    and $y_3 = (1,0)$. Set $\psi^t = (0,t,0)$. Then,
    $$\frac{\partial G_1}{\partial \psi_3}(\psi^t) =
    \Haus^1(K \cap \Lag_1(\psi^t) \cap \Lag_3(\jocelyn{\psi^t})) = \begin{cases}
        0 &\hbox{when } t > \frac{\jocelyn{-6}}{4}\\
      1 &\hbox{when }  t < \frac{\jocelyn{-6}}{4},
    \end{cases}$$
    thus showing that $G$ is not globally $\Class^1$.
    \item The regularity hypothesis would never be satisfied when one of
      the simplex is one-dimensional, thus explaining the first
      hypothesis in our definition of regular simplicial measure
      (Def. \ref{def:reg-simplicial-measure}). Note also that this lack of genericity translates
      into a lack of regularity for $G$. Indeed, take $\mu$ the
      uniform measure over a segment $[a,b]$. Then, the partial derivative
      $$ \frac{\partial G_i}{\partial \psi_j}(\psi) =
      \Haus^0(\Lag_{i}(\psi) \cap \Lag_j(\psi) \cap [a,b]) =
      \Card(\Lag_{i}(\psi) \cap \Lag_j(\psi) \cap [a,b])),$$ can only
      take values in $\{0,1\}$ and must be discontinuous or constant.
\end{itemize}
\end{remark}

The end of this section if devoted to the proof of Theorem~\ref{thm:regularity}. We first remark that by linearity of the integrals in the definition
of $G$ with respect to $\mu$, the theorem will hold for a simplicial
measure if it holds for any measure with density supported on a
simplex. We therefore let $\sigma$ be a $k$-dimensional simple of
$\Rsp^d$ and $\mu = \mu_\sigma$ be a measure on $\sigma$ with
continuous density $\rho_\sigma:\sigma\to\Rsp$ with respect to the
$k$-dimensional Hausdorff measure on $\sigma$. We also introduce 
\begin{equation} \label{eq:GradPhi-sigma}
 G_{\sigma,i}(\psi) := \int_{\Lag_i(\psi)\cap \sigma} \rho_\sigma(x) \dd\Haus^{k}(x) \dd x.
\end{equation}
The following lemma will be used to compute the first derivatives of
the function $G_{\sigma,i}$.

\begin{lemma}\label{lemma:transversality} 
Let $\rho:\Rsp^k\to\Rsp$ be a continuous function on $\Rsp^k$ and let
$z_1,\cdots,z_N \in \Rsp^k$ be vectors whose conic hull is
$\Rsp^k$ (i.e. $\forall x\in \Rsp^k,\exists
\lambda_1,\hdots,\lambda_\jocelyn{N} \geq 0$ s.t. $x = \sum_i \lambda_i z_i$). Given
$\lambda\in\Rsp^k$, define
\begin{align}  \label{eq:K-loc}
&\Kloc(\lambda) := \{ x \in \Rsp^k \mid \forall i \in \{1, \ldots, N \},~
\sca{x}{z_i} \leq \lambda_i \},\\
\label{eq:G-loc}
&\Gloc(\lambda) := \int_{\Kloc(\lambda)} \rho(x) \dd\Haus^k(x).
\end{align}  Then,
\begin{itemize}
\item Assume that the $z_i$ are non-zero. Then, the
  function $\Gloc$ is continuous.
\item Assume that all the vectors $z_i$ are pairwise independent
  (i.e. not collinear, implying in particular that they are
  non-zero). Then $\Gloc$ has class $\Class^1$ and its partial
  derivatives are 
\begin{equation}
\frac{\partial \Gloc}{\partial \lambda_i}(\lambda) = \frac{1}{\nr{z_i}}
\int_{\jocelyn{\Kloc(\lambda)} \cap \{ x\mid \sca{x}{z_i} = \lambda_i\}} \rho(x) \dd\Haus^{k-1}(x)
\end{equation}
\end{itemize}
\end{lemma}
\begin{proof}  Let $e_1,\hdots,e_N$ be the canonical basis of
  $\Rsp^N$.

 \noindent\textbf{Step 0.} Note that, because the conic hull of the
  $z_i$ equals $\Rsp^k$, the polytope $\Kloc(\lambda)$ is always
  compact. Moreover, one easily sees that if $\lambda \leq \lambda'$
  (coordinate-wise), one has $\Kloc(\lambda)\subseteq \Kloc(\lambda')$. This
  implies that
  \begin{equation} \forall R\geq 0,~~ \exists C_R\subseteq \Rsp^d \hbox{ compact s.t. }
    \forall \lambda' \in \Rsp^N\  \max_i |\lambda'_i - \lambda_i| \leq R \Rightarrow \Kloc(\lambda')
    \subseteq C_R. \label{eq:CR}
  \end{equation}
  We now sketch how to prove the continuity of the function $\Gloc$ near
  any $\lambda\in\Rsp^N$. Let $t\in [-R,R]$. We can assume that $t\geq 0$. First, note that the symmetric difference
  $\Kloc(\lambda) \Delta \Kloc(\lambda + t e_i)$ is contained in a slab, or more precisely
  $$ \Kloc(\lambda) \Delta \Kloc(\lambda + t e_i) \subseteq
  C_R\cap \{ x \in \Rsp^d \mid \sca{x}{z_i} \in [\lambda,\lambda+t]\}, $$
and that the width of the slab is $t/\nr{z_i}$. This gives
  \begin{align*}
    \abs{\Gloc(\lambda) - \Gloc(\lambda+te_i)}
    \leq \int_{\Kloc(\lambda) \Delta \Kloc(\lambda + t e_i)} \rho(x) \dd x 
    \leq \left[\frac{\diam(C_R)^{d-1}\max_{C_R} \abs{\rho}}{\nr{z_i}}\right] t
  \end{align*}
  A similar bound obviously exist for $t\leq 0$. Using this estimate on each coordinate axis, one obtains the
  continuity of $\Gloc$ (and in fact, this proof even shows that $\Gloc$ is
  locally Lipschitz). This proves the first statement.
  
\noindent\textbf{Step 1.} We now prove the second statement, and assume that
  $\rho$ is continuous and the $z_i$ are pairwise independent. Fix
  some index $i_0 \in \{1,\hdots,N\}$ and take $\lambda \in \Rsp^N$. We consider the convex set $L
  := \{ x \in \Rsp^k \mid \forall i \neq i_0\  \sca{x}{z_{i}} \leq \lambda_i\}$. For any $t\geq 0$, using the function $u: x \in \Rsp^k\mapsto \sca{x}{z_{i_0}} - \lambda_{i_0}$, one has $\Kloc(\lambda + t e_{i_0}) \setminus \Kloc(\lambda) = L \cap u^{-1}([0,t])$.  Applying the co-area formula with the function $u$ whose gradient is $\nabla u = z_{i_0}$, we can evaluate  the slope
\begin{align}
\frac{1}{t} (\Gloc(\lambda + t e_{i_0}) - \Gloc(\lambda))
&= \frac{1}{t}\int_{L \cap u^{-1}([0,t])} \rho(x)d\Haus^k(x)  \notag \\
&= \frac{1}{t}\int_{0}^{t} \int_{L \cap u^{-1}(s)} \frac{\rho(x)}{\nr{z_{i_0}}} \dd\Haus^{k-1}(x) \dd s \notag \\
&= \frac{1}{t}\int_{0}^{t} g_{i_0}(\lambda + s e_{i_0}) \dd s \label{eq:Glambda}
\end{align}
where we have set
$$ g_{i_0}(\overline{\lambda}) := \int_{\Kloc_{i_0}(\overline{\lambda})}
\frac{\rho(x)}{\nr{z_{i_0}}} \dd \Haus^{k-1}(x) \quad \hbox{ with }
\Kloc_{i_0}(\overline{\lambda}) = \{ x\in \Kloc(\overline{\lambda}) \mid \sca{x}{z_{i_0}} =
\overline{\lambda}_{i_0} \}$$ Note that by construction, $\Kloc_{i_0}(\overline{\lambda})$ is
the facet of $\Kloc(\overline{\lambda})$ with exterior normal
$z_{i_0}/\nr{z_{i_0}}$. Assume for now that we are able to prove that
the functions $g_{i_0}$ are continuous. Then, by the fundamental theorem
of \jocelyn{calculus} and by Equation~(\ref{eq:Glambda}) one has $\frac{\partial
  \Gloc}{\partial \lambda_{i_0}}(\lambda) = g_{i_0}(\lambda).$ Since we
have assumed that $g_{i_0}$ is continuous, this shows that the
function $\Gloc$ has continuous partial derivatives and is therefore
$\Class^1$, and gives the desired expression for its partial
derivatives.

\noindent\textbf{Step 2.} Our goal is now to establish the continuity of the
function $g_{i_0}$. In order to do that, we will parameterize the
facet $\Kloc_{i_0}(\lambda)$ using the hyperplane $V = \{z_{i_0}\}^\perp$
and $\Pi_V$ the orthogonal projection on this hyperplane. Then,
decomposing $x\in \Kloc_{i_0}(\lambda)$ as $\Pi_V(x) +
\lambda_{i_0}\frac{z_{i_0}}{\nr{z_{i_0}}^2}$ we get
$$ g_{i_0}(\lambda) = \frac{1} {\nr{z_{i_0}}} \int_{\Pi_V
  (\Kloc_{i_0}(\lambda))} \rho\left(y + \lambda_{i_0}
\frac{z_{i_0}}{\nr{z_{i_0}}^2}\right) \dd \Haus^{k-1}(y) $$
By compactness, $\rho$ is uniformly continuous on $C_R$,
where $C_R$ is defined in Eq.~(\ref{eq:CR}): there exists a function
$\omega_R: \Rsp^+\to\Rsp^+$ satisfying $\lim_{r\to 0} \omega_R(r) = 0$
and such that for all $x,y\in C_R,$ $\abs{\rho(x) - \rho(y)} \leq
\omega_R(\nr{x-y})$. Using the function $\rho_{\lambda}(y) := \rho(y + \lambda_{i_0} z_{i_0}/\nr{z_{i_0}}^2)$ and the notation $\tilde{\Kloc}_{i_0}(\lambda) = \Pi_V(\Kloc_{i_0}(\lambda))$, one has for every $\lambda'$
\begin{align}
  &\nr{z_{i_0}} \abs{g_{i_0}(\lambda) - g_{i_0}(\lambda')}  \notag \\ 
  & = \abs{\int_{\tilde{\Kloc}_{i_0}(\lambda)} \rho_\lambda(y)\dd \Haus^{k-1}(y) - \int_{\tilde{\Kloc}_{i_0}(\lambda')} \rho_{\lambda'}(y)\dd \Haus^{k-1}(y) }  \notag\\ 
  &\leq  \abs{\int_{\tilde{\Kloc}_{i_0}(\lambda)} (\rho_\lambda(y) - \rho_{\lambda'}(y)) \dd \Haus^{k-1}(y) }
  \notag\\  &
  \quad + \abs{\int_{\tilde{\Kloc}_{i_0}(\lambda)} \rho_{\lambda'}(y)\dd \Haus^{k-1}(y) 
  - \int_{\tilde{\Kloc}_{i_0}(\lambda')} \rho_{\lambda'}(y)\dd \Haus^{k-1}(y) } 
   \notag\\ 
   \label{ineq:g}   
   \end{align}
 Suppose now that $\max_i|\lambda_i - \lambda_i'| \leq R$. Then the first term of the right hand side term is bounded by $\Haus^{k-1}(\Pi_V(C_R)) \omega_R(\abs{\lambda_{i_0} - \lambda'_{i_0}}/\nr{z_{i_0}})$ which tends to zero when $\lambda'$ tends to $\lambda$. For the second term, we note that
\begin{align*}
  \tilde{\Kloc}_{i_0}(\lambda) &= \{ y \in V \mid
  \forall i\neq i_0,~\sca{y + \lambda_{i_0} z_{i_0}/\nr{z_{i_0}}^2}{z_i} \leq \lambda_i\} \\
  &= \{ y \in V \mid
  \forall i\neq i_0,~\sca{y}{\tilde{z}_i} \leq \lambda_i - \lambda_{i_0} \sca{z_i}{z_{i_0}}/\nr{z_{i_0}}^2\},
\end{align*}
where we have set $\tilde{z_i} = \Pi_{V}(z_i) = \Pi_{\{z_{i_0}\}^\perp}
(z_{i})$. The assumption that $z_i$ and $z_{i_0}$ are independent
implies that the vectors $\tilde{z}_{i}$ are non-zero. We conclude
using the first part of the Lemma that the function 
$$
\lambda \mapsto \int_{\tilde{\Kloc}_{i_0}(\lambda)} \rho_{\lambda'}(y)  \dd \Haus^{k-1}(y)
$$
is continuous. Using the inequality \eqref{ineq:g}, we see that
$\lim_{\lambda'\to\lambda} g_{i_0}(\lambda') = \lambda$. This shows
that $g_{i_0}$ is continuous and concludes the proof of the lemma.

\end{proof}



We will also use the following easy consequence of the genericity
hypothesis.

\begin{lemma}\label{lemma-genericity} Assume $\{ y_1, \ldots, y_N \} \subset \Rsp^d$ is in
generic position with respect to a $k$-dimensional simplex
$\sigma=[x_0,\hdots,x_{k}]$ and let $H = \vect(x_1 - x_0, \hdots, x_k
- x_0).$ Then,
\begin{itemize}
\item For every pairwise distinct $i,j,l \in\{1,\hdots,n\}$, the
  vectors $z_1 = \pi_H(y_j - y_i)$ and $z_2 = \pi_H(y_l - y_i),$ where
  $\pi_H$ is the orthogonal projection on $H$, are not collinear.
\item For every distinct $i,j\in \{1,\hdots,n\}$, the vector
  $\pi_H(y_j - y_i)$ is not perpendicular to any of the
  $(k-1)$-dimensional facets of $\sigma$.
\end{itemize}
\end{lemma}
\begin{proof}By the genericity condition of Definition \ref{def:Generic},
    $\jocelyn{\{y_j-y_i\}}^\perp \cap H$ is of dimension $k-1$. Furthermore, for a
    vector $u\in \jocelyn{\{y_j-y_i\}}^\perp \cap H$, one has  $\sca{u}{y_j-y_i}=0$ and
    $\sca{u}{\jocelyn{z_1}}=0$ which implies that $\jocelyn{\{y_j-y_i\}}^\perp \cap
    H=\{z_1\}^\perp \cap H$. Similarly, one has $\jocelyn{\{y_l-y_i\}}^\perp \cap
    H=\{z_2\}^\perp \cap H$. If $z_1$ and $z_2$ are collinear, then
    $\jocelyn{\{y_j-y_i,y_l-y_i\}}^\perp \cap H=(\jocelyn{\{y_j-y_i\}}^\perp
    \cap H) \cap (\jocelyn{\{y_l-y_i\}}^\perp \cap H)$ is of dimension $k-1$ which contradicts the genericity condition. The proof of the second item is straightforward.
\end{proof}

\begin{proof}[Proof of Theorem \ref{thm:regularity}.] Our goal is to show
that $G_{i,\sigma}$ (defined in \eqref{eq:GradPhi-sigma}) is
$\Class^1$--regular and to compute its partial derivatives.  From now
on, we fix some index $i_0 \in \{1,\hdots,N\}$. Reordering indices if
necessary, we assume that $i_0 = N$. We want to apply
Lemma~\ref{lemma:transversality}, and for that purpose we are first
going to rewrite $\Lag_i(\psi) \cap \sigma$ under the form
\eqref{eq:K-loc}. Denote $H$ the $k$-dimensional affine space spanned
by $\sigma$; translating everything if necessary, we can assume that
$H$ is a linear subspace of $\Rsp^d$. A simple calculation shows that
the intersection of the $N$th power cell with $H$ is given by
\begin{equation*}
\Pow_N(\psi) \cap H = \{ x \in H
\mid \forall i \in \{1, \ldots, N-1\}, \sca{x}{z_i} \leq \lambda_i \},
\end{equation*}
where $\lambda_i = \frac{1}{2} (\nr{y_i}^2 + \psi_i - (\nr{y_N}^2 +\psi_N))$
and $z_i$ is the orthogonal projection of $y_i - y_{N}$ on
$H$.  Since $\sigma$ is a $k$-dimensional simplex, it can be written
as the intersection of $k+1$ half-spaces of $H$, i.e. 
$\sigma = \{ x \in H \mid \forall j \in \{N, \ldots, N+k\},
\sca{x}{z_j} \leq 1 \}$ for some non-zero vectors $z_i$ of
$H$. Combining these two expressions, one
gets 
$$ 
\Lag_\jocelyn{N}(\psi) \cap \sigma = \{ x \in H \mid \forall i \in \{1, \ldots,
N+k \},~ \sca{x}{z_i} \leq \lambda_\jocelyn{i} \}. 
$$ 
\jocelyn{where $ \lambda_i = 1 $ for $ i \in \{N, \ldots, N+k\} $.}



We will now show that the assumptions of Lemma
\ref{lemma:transversality} are satisfied. Since $\sigma$ is a
nondegenerate simplex, $z_i\neq 0$ for every \jocelyn{$i\geq N$} and the
vectors $z_i,z_j$ for $i\neq j$ and $i,j\geq N$ are pairwise
independent. From the first genericity property of Lemma
\ref{lemma-genericity}, we know that $z_i = \Pi_H(y_i - y_N)$ and $z_j
= \Pi_H(y_j - y_N)$ are independent ($i\neq j$ and $i,j < N$). From
the second genericity condition, we also know that $z_i,z_j$ are
independent when $i\neq j$ and $i<N$ and $j\geq N$. In order to apply
Lemma~\ref{lemma:transversality} we need to extend the continuous
density $\rho_\sigma: \sigma\subseteq H\to\Rsp$ into a continuous
density $\rho: H\to\Rsp$. Since $\sigma$ is convex, this can be easily
done using the projection map $\Pi_\sigma:H\to\sigma$, and by setting
$\rho(x) = \rho_\sigma(\Pi_\sigma(x))$. Then, $\rho$ is continuous as
the composition of two continuous maps (recall that since $ \sigma $ is convex, the projection
$\Pi_\sigma$ is $1$-Lipschitz). With these constructions one has
$$ G_{\sigma,N}(\psi) = \Gloc(A(\psi)),$$
where $A:\Rsp^{N} \to\Rsp^{N+k}$ is the affine map
$$ A(\psi) := \left(\frac{1}{2} (\nr{y_1}^2 + \psi_1 - (\nr{y_N}^2
+\psi_N)), \hdots, \frac{1}{2} (\nr{y_{N-1}}^2 + \psi_{N-1} -
(\nr{y_{N}}^2 +\psi_N), 1, \hdots,1\right)$$ with $k+1$ trailing ones. By
Lemma~\ref{lemma:transversality}, $\Gloc$ has class $\Class^1$, and the
expression above shows that $G_{\sigma,N}$ is also
$\Class^1$. Moreover, denoting $A = (A_1,\hdots,A_{N+k}),$ one gets
\begin{align*}
 \forall i\neq N,~~ \frac{\partial G_{\sigma,N}}{\partial \psi_i} (\psi) 
&= \sum_{1\leq j\leq N+k} \frac{\partial A_j}{\partial{\psi_i}}(\psi) \frac{\partial \Gloc}{\partial \lambda_j}(A(\psi))\\
&= \frac{1}{2} \frac{\partial \Gloc}{\partial \lambda_i}(A(\psi)) \\
&= \frac{1}{2\nr{z_i}} \int_{\Kloc(A(\psi))\cap\{x\in H\mid \sca{x}{z_i} = \lambda_i\}} \rho_\sigma(x) \dd\Haus^{k-1}(x) \\
&= \frac{1}{2\nr{z_i}} \int_{\Lag_{i,N}(\psi) \cap \sigma} \rho_\sigma(x) \dd\Haus^{k-1}(x),
\end{align*}
thus establishing the first formula in \eqref{eq:HessPhi}. The second
formula in this equation deals with the case $i=N$, and follows from a
similar computation and from the expression
$$\left(\frac{\partial A_j}{\partial \psi_N}(\psi)\right)_{1\leq j\leq
  N+k} = \left(-\frac{1}{2},\hdots,-\frac{1}{2}, 0,\hdots,
0\right), $$ with $k+1$ trailing zeros. We have therefore established
the theorem when $\mu = \mu_{\sigma}$. The case where $\mu =
\sum_{\sigma \in \Sigma} \mu_{\sigma}$ is a simplicial measure follows
by linearity.
\end{proof}

%% file: inc/3-concavity.tex
\section{Strict monotonicity of $G$}\label{sec:motonicity}
As mentioned in Section \ref{sec:transport-plan}, the second ingredient needed for the proof
of the convergence of the damped Newton's algorithm is a motonicity property of $ G $.
This property relies heavily on the ``strong
connectedness'' of the support of $\mu$ assumed in the third item of
Def.~\ref{def:reg-simplicial-measure}.
We denote by $\cstperp  = \{ v\in\Rsp^Y \mid \sum_{1\leq i\leq N} v_i = 0\}$ the orthogonal of the constant functions on $Y$. 

\begin{theorem} \label{thm:strict-concavity}
  Let $\mu$ be a regular simplicial measure and assume that
  $y_1,\hdots,y_N$ is generic with respect to the support of $\mu$
  (Def.~\ref{def:Generic}). Then $G$ is strictly monotone
  in the sense
  that
    $$ \forall \psi \in K^{+},~ \forall v \in \cstperp\setminus\{0\},~ \sca{\D G(\psi) v}{v} < 0. $$
\end{theorem}

\begin{figure}[h]
    \centering
    \includegraphics[scale=0.33]{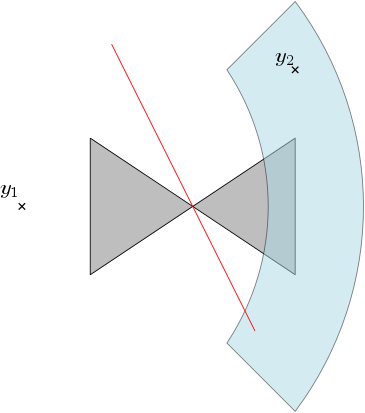}
    \caption{Simplex soup where the set of points $ y_1, y_2 $
        such that $ \mu(\Lag_{1, 2}(\psi)) = 0 $ has not a zero measure.}
    \label{fig:simplicial-complex-problem}
\end{figure}

\begin{remark}
Let us illustrate the fact that the connectedness of $K$ is not sufficient (i.e. why
we require that it is impossible to disconnect the support $K$ of
$\mu$ by removing a finite number of points). Consider the case where $
K $ is made of the two 2-dimensional simplices embedded in $\Rsp^2$,
and displayed in grey in Figure
\ref{fig:simplicial-complex-problem}. We assume that $\mu$ is the
restriction of the Lebesgue measure to $K$ and that $Y = \{y_1,y_2\}$.
Then, the matrix of  the differential of $G$ at $\psi$ is the $2$-by-$2$ matrix
given by
$$ \D \jocelyn{G}(\psi) = \begin{pmatrix} a & -a \\ -a &
  a \end{pmatrix} \hbox{ where } a = \frac{1}{2\nr{y_1-y_2}} \Haus^1(\Lag_{1,2}(\psi) \cap
K). $$ If we fix $ y_1\in\Rsp^2 $, it is easy to see that for any $ y_2 $ in
the blue domain, there exists weights $ \psi_1 $ and $ \psi_2 $ such
that the interface $\Lag_{1, 2}(\psi)$ (in red) passes through the
common vertex between the two simplices, thus implying that
$a= 0$, hence $\D G(\psi)=0$. In such setting, $G$ is not
strictly monotone, the conclusion of Theorem~\ref{thm:strict-concavity}
does not hold.
\end{remark}

The end of this section is devoted to the proof of Theorem~\ref{thm:strict-concavity}. 

\subsection{Preliminary lemmas} \label{subsec:technical-lemmas}

With a slight abuse, we call tangent space to a convex set $K$ the
linear space $\vect(K - x)$ for some $x$ in $K$ (this space is
independent of the choice of $x$). We denote $\relint(K)$
the \emph{relative interior} of a convex set $K \subseteq \Rsp^d$ and
we call dimension of $K$ the dimension of the affine space spanned by
$K$.


\begin{lemma} \label{lemma:ri} 
Let $e,f$ be convex sets and $E$ and $F$ their tangent spaces. Assume
that $\relint(f) \cap \relint(e)\neq \emptyset$. Then,
$$ \dim(e\cap f) = \dim(E\cap F).$$
\end{lemma}

\begin{proof}
Let $G$ be the tangent space to $e\cap f$, so that $\dim(e\cap f)
= \dim(G)$. It suffices to show that $G = E\cap F$ to prove that
$\dim(e\cap f) = \dim(E\cap F)$. The inclusion $G\subseteq E\cap F$
holds without hypothesis (a tangent vector to $e\cap f$ is always both
a tangent vector to $e$ and to $f$). For the reciprocal inclusion,
consider $x\in \relint(e)\cap\relint(f)$ and $v\in E\cap F$. Then, by
definition of the relative interior, for $t$ small enough one has
$x+tv \in e$ and $x+tv\in f$, i.e. $x+tv\in e\cap f$, so that $tv$
belongs to $G$. This shows $G\subseteq E\cap F$ and concludes the
proof.
\end{proof}

\begin{lemma} \label{lemma:ri-generic}
Let $f\subseteq f'$ and $e$ be three convex sets of $\Rsp^d$, and
    $F \subseteq F'$ and $E$ be their tangent spaces. Assume that
    \begin{itemize}
\item $\relint(f) \cap \relint(e)\neq \emptyset$ ;
\item  $\dim(F') = \dim(F)+1$ and $\dim(E\cap F') = \dim(E\cap F)+1$.
\end{itemize}
Then $\dim(e\cap f') = \dim(e\cap f) + 1$.
\end{lemma}
\begin{proof}
Let us first show that
$\relint(e)\cap \relint(f') \neq \emptyset$
.We consider
a basis $e_1,\hdots,e_n$ of $F$ and a vector $e_{n+1} \in E \cap F'$
such that $E \cap F' = (E \cap F) \oplus \Rsp e_{n+1}$ and $F' =
F \oplus \Rsp e_{n+1}$. Let $x_0$ be a point in the intersection
$\relint(f) \cap \relint(e)$, which we assumed non-empty. There exists
$\eps>0$ such that $\Delta := \conv(\{ x_0 \pm \epsilon e_i \mid 1\leq i\leq n \}) \subseteq f.$ Using the assumption that $F'$ is the
tangent space to $f'$, we know that there exists a point $y \in \jocelyn{f'}$
such that $v = y - x_0 \in F'\setminus F$. Consider the convex sets
$\Delta_\pm$ spanned by $\Delta$ and one of the points $x_0 \pm v$,
$\Delta_\pm = \conv(\Delta \cup \{ x_0 \pm v\})$
. The convex set
$\Delta_+ \cup \Delta_-$ is a neighborhood of $x_0$, meaning that
there exists $t\neq 0$ such that $x_\pm := x_0 \pm
te_{n+1} \in \relint(\Delta_\pm)$. Assume for instance
$x_+ \in \relint(\Delta_+) \subseteq f'$. Since $\Delta_+$ has the
same dimension as $f'$, one has
$x_+ \in \relint(\Delta_{+}) \subseteq \relint(f')$ and by a standard
property of the relative interior one has $(x_0, x_+] = (x_0,
x_0+te_{n+1}] \subseteq \relint(f')$. Finally, since $x_0$ belongs to
the relative interior of $e$ and $e_{n+1} \in E$, the segment $(x_0,
x_0+te_{n+1}]$ must intersect the relative interior of $e$, proving
that $\relint(e)\cap \relint(f') \neq \emptyset$.

Then using Lemma~\ref{lemma:ri}, we have $\dim(e\cap f) = \dim(E\cap F)$ and 
$\dim(e\cap f') = \dim(E \cap F') = \dim(e\cap f) + 1$.
\end{proof}

\subsection{Proof of the strict motonicity of $ G $} \label{subsec:proof-monotonicity}
This theorem will follow using standard arguments, once one has
established the connectedness of the graph induced by the \jocelyn{Jacobian}
matrix. Let $\psi\in K^+$, $H := \D G(\psi)$ and consider the
graph $ \jocelyn{\mathcal{G}} $ supported on the set of vertices $ V = \{1, \ldots, N\}$
and with edges
$$
E(\jocelyn{\mathcal{G}}) := \{ (i,j) \in V^2 \mid  i\neq j \hbox{ and }H_{i, j}(\psi) > 0  \}.
$$

\begin{lemma}
If $\Lag_{i,j}(\psi)$ intersects some $k$-dimensional simplex $\sigma\in\Sigma$, then
the intersection is either a singleton or has dimension $k-1$.
\end{lemma}

\begin{proof}
Denote $\sigma = [x_0,\hdots,x_k]$
and assume that $m = \dim(\Lag_{i,j}(\psi) \cap \sigma) \geq 1$.
Consider a $p$-dimensional facet $f=[x_{j_0},\hdots,x_{j_p}]$ of
$\sigma$ and a facet
$\Lag_{i_0,\hdots,i_\ell}(\psi) \jocelyn{= \bigcap_{k=0}^\ell
    \Lag_{i_k}(\psi)} $ of $\Lag_{i,j}(\psi)$ (we take $i_0=i$ and $i_1=j$) 
such that $\dim(\Lag_{i_0,\hdots,i_\ell}(\psi) \cap f) = m$ and assume
that both facets are minimal for the inclusion. It is easy to see that
this minimality property implies that the relative interiors of $f$
and $\Lag_{i_\jocelyn{0},\hdots,i_\ell}(\psi)$ must intersect each other. With Lemma~\ref{lemma:ri}, this
ensures that
\begin{align} \label{eq:mdim}
  m &= \dim(\Lag_{i_0,\hdots,i_\ell}(\psi)\cap f) \\
  &= \dim(\jocelyn{\{y_{i_1} - y_{i_0},\hdots,y_{i_\ell} - y_{i_0}\}}^\perp \cap
\vect(\jocelyn{x_{j_1} - x_{j_0},\hdots,x_{j_p} - x_{j_0}})) = p-\ell,
\end{align}
where we used the genericity property (Def~\ref{def:Generic}) to get
the last equality.  We now prove that $p=k$ and $\ell=1$ by
contraction. If we assume that $p<k$, there exists $j_{p+1} \in
\{1,\hdots,k\}$ distinct from $\{{j_0},\hdots,{j_p}\}$. Set $e =
\Lag_{i_0\hdots,i_\ell}(\psi)$, $f = [x_{j_0},\hdots,x_{j_{p}}]$ and
$f' = [x_{j_0},\hdots,x_{j_{p+1}}]$. The genericity hypothesis allows
us to apply Lemma~\ref{lemma:ri-generic}. The conclusion of the lemma
is that $\dim(\Lag_{i_\jocelyn{0},\hdots,i_\ell}(\psi)\cap f') = p+1-\ell> m$,
which violates the definition of $m$. By contradiction one must have
$p=k$. With the same arguments (removing a point \jocelyn{$y_{i_n}$ for some $ n
    \in \{0, \ldots, \ell\} $ different from $ y_i $ and $ y_j$} from the
list if $i_\ell \geq 1$) we can see that necessarily $\ell= 1$. With
\eqref{eq:mdim} we get $m=k-1$, thus concluding the proof of the
lemma.
\end{proof}

\begin{lemma} The graph $\jocelyn{\mathcal{G}} $ is connected.
\end{lemma}

\begin{proof}
  Consider the finite set
  $$S := \{ x \in \Rsp^d \mid \exists \sigma \in \Sigma, \exists i
  \neq j\in\{1,\hdots,N\},~ \Lag_{i,j}(\psi) \cap \sigma = \{x\} \}.$$
  For any simplex $\sigma\in\Sigma$, denote $\sigma^* =
  \sigma\setminus S$, and let $K^* = K\setminus S$. By definition of a
  regular simplicial measure (Def.~\ref{def:reg-simplicial-measure}),
  we know that $K^*$ is connected.  Let $C=\{i_1,\hdots,i_c\}$ be a
  connected component of the graph $\jocelyn{\mathcal{G}} $, and
  define $L = \bigcup_{i\in C} \Lag_{i}(\psi)$ and $L' =
  \bigcup_{i\not\in C} \Lag_i(\psi)$.

  \noindent\textbf{Step 1} We first show that for any simplex
  $\sigma\in \Sigma$, one must have either $\sigma^* \subset
  \interior{L}$ \emph{or} $\sigma^* \subset \interior{\Rsp^d\setminus
    L}$. For this, it suffices to prove that for any
  $\sigma\in\Sigma$, $\sigma^*\cap \partial {L} = \emptyset$. We argue
  by contradiction, assuming the existence of a point $x \in
  \partial{L} \cap \sigma^*$. Then, by definition of $\partial L$,
  there exists $i\in C$ and $j\not\in C$ such that $x\in \Lag_{i,j}(\psi)$. Since $x\in \sigma^*$, we know that $x$ does not belong to $S$. This implies that  $\Lag_{i,j}(\psi)\cap\sigma$ cannot be a singleton.
  %
  By the previous Lemma, this gives $\dim(\sigma\cap
  \Lag_{i,j}(\psi))=d_\sigma-1$ so that
$$ H_{ij}(\psi) = \jocelyn{\const(y_i, y_j)} \int_{\sigma\cap
  \Lag_{i,j}(\psi)} \rho_\sigma(x) \dd\Haus^{d_\sigma-1}(x) > 0.$$
This shows that $i$ and $j$ are in fact adjacent in the graph $\jocelyn{\mathcal{G}}$ and
contradicts $j\not\in C$.

\noindent  \textbf{Step 2} We now prove that $C$ is equal to
$\{1,\hdots,N\}$ by contradiction. 
 We group the simplices $\sigma \in \Sigma$ according to whether
$\sigma^*$ belongs to $\interior{L}$ or to $\interior{\Rsp^d\setminus   L}$.  
  {The sets $K^*_i$ are open for the topology induced on $K^*$ because $K^*_1=\interior{L} \cap K^*$ and $K^*_2=\interior{L'} \cap K^*$. Since they are also non empty, this    violates the connectedness of $ K^* $.}
We can conclude that $C=\{1,\hdots,N\}$, i.e. $\jocelyn{\mathcal{G}}$
is connected.
\end{proof}

\begin{proof}[Proof of Theorem \ref{thm:strict-concavity}]
  First note that the matrix $H$ is symmetric and therefore
  diagonalizable in an orthonormal basis. Gershgorin's circle theorem
  immediately implies that the eigenvalues of the matrix are
  negative. The theorem will be established if we are able to show
  that the nullspace of $H$ (i.e. the eigenspace corresponding to the
  eigenvalue zero) is the $1$-dimensional space generated by constant
  functions. \jocelyn{The computations presented here are similar to the ones in
      \cite[Lemma 3.3]{carlier2010knothe}.}
  Consider $v$ in the nullspace and let $i_0$ be an index where $v$
  attains its maximum, i.e. $i_0 \in \argmax_{1 \leq i \leq n} v_i $.
  Then, using $Hv=0,$
    \begin{align*}
        0 = (H v)_{i_0} = H_{i_0, i_0} v_{i_0} + \sum_{i \neq i_0}
        H_{i, i_0} v_i &\leq H_{i_0, i_0} v_{i_0} + \sum_{i \neq i_0}
        H_{i, i_0} v_{i_0} \\ &= H_{i_0, i_0} v_{i_0} + \left( \sum_{i
          \neq i_0} H_{i, i_0} \right) v_{i_0} = 0.
    \end{align*}
The inequality follows from $v_i \leq v_{i_0}$ and from $H_{i,i_0}\geq
0$, while the third equality comes from $H_{i_0, i_0} = -\sum_{i \neq
  i_0} H_{i_0, i}.$ This allows us to write $v_{i_0}$ as convex
combination of values $v_i\leq v_{i_0},$
$$ v_{i_0} = \sum_{i\neq i_0} \frac{H_{i,i_0}}{\sum_{j\neq i_0}
  H_{\jocelyn{j},i_0}} v_i.$$ This means that for all vertex $i$ adjacent to
$i_0$ in the graph $\jocelyn{\mathcal{G}}$ (so that $H_{i,i_0}\neq 0$), one must have $v_i
= v_{i_0}$. In particular, the function $v$ also attains its maximum
at $i$. By induction and using the connectedness of the graph $\jocelyn{\mathcal{G}}$,
this shows that $v$ has to be constant, i.e. $\Ker(H) =
\vect(\{\mathrm{cst}\})$.
\end{proof}


%% file: inc/5-convergence.tex
\section{Convergence analysis}
\label{sec:convergence}

In this section, we show the convergence of a damped Newton algorithm
for a general function $G:\Rsp^N\to\Rsp^N$ that satisfies some
regularity and strict monotonicity conditions. As a direct
consequence, using the results of Sections \ref{sec:regularity} and
\ref{sec:motonicity}, we show the convergence with a linear speed of
the damped Newton algorithm to solve the non-linear
equation~\eqref{eq:MA}. We denote by $\mathcal{P}_N$ the set of
$\nu=(\nu_1,\cdots,\nu_N) \in \Rsp^N$ that satisfies $\nu_i \geq 0$
and $\sum_i \nu_i =1$. For a given function $G:\Rsp^N \to
\mathcal{P}_N$ and $\eps >0$, we define the set
$$
\K^{\epsilon}:= \left\{ \psi \in \Rsp^N \mid \forall i,~ G_i(\psi) \geq \epsilon \right\},
$$
where $ G(\psi) = (G_i(\psi))_{1 \leq i \leq  N}$. We then have the following
proposition, which is an adaptation to our setting of Theorem 1.5 in \cite{kitagawa2016newton} and
Proposition 2.10 in \cite{mirebeau2015discretization}.

\begin{proposition}     \label{prop:convergence-newton}
    Let $G:\Rsp^N \to \mathcal{P}_N$ be a function which is invariant
    under the addition of a constant\jocelyn{, i.e. a multiple of $(1, \ldots,
        1) \in \Rsp^N $,} and $\eps>0$. We assume the following properties:
    \begin{enumerate}
        \item (Compactness) For every $a\in \Rsp$, the following set is compact:
        $$ \K^{\epsilon}_a := \K^{\epsilon} \cap \left\{ \psi \in \Rsp^N \mid
            \sum_{i=1}^N \psi_i = a\right\} = \left\{ \psi \in \Rsp^N \mid
            \forall i,~ G_i(\psi) \geq \epsilon \mbox{ and } \sum_{i=1}^N \psi_i = a\right\}.
        $$
                \item ($\Class^1$ regularity) The function $G$ is of class $
            \Class^1 $ on $
            \K^{\epsilon}$.
        \item (Strict monotonicity) We have:
            $$ \forall \psi \in \K^{\epsilon},~ \forall v \in \cstperp \setminus
            \{0\},~ \sca{DG(\psi) v}{v} < 0 $$
    \end{enumerate}
   Then Algorithm~\ref{algo:newton} converges with linear speed. More
   precisely, if $ \nu \in \mathcal{P}_N$ and $ \psi_0 \in \Rsp^N $
   are such that $ \epsilon_0 = \frac{1}{2} \min \left( \min_i
   G_{i}(\psi_0), \min_i \nu_i \right) > 0$, then the iterates $
   (\psi^k) $ of Algorithm \ref{algo:newton} satisfy the following
   inequality, where $ \tau^* \in (0,1] $ depends on $ \epsilon_0$:
    $$ \nr{G(\psi^{k+1}) - \nu} \leq \left( 1 - \frac{\tau^\star}{2} \right)
    \nr{G(\psi^{k}) - \nu}, $$

\end{proposition}

\begin{proof}
Let $ \nu \in \mathcal{P}_N$ and $ \psi^0 \in \Rsp^N $ such that $
\jocelyn{\epsilon} = \frac{1}{2}  \min \left( \min_i G_{i}(\psi^0), \min_i \nu_i \right)
$ is positive.
 We put $a=\sum_{i=1}^N \psi^0_i$. Since $\K^{\epsilon}_a$ is a compact set,
 the continuous
    map $\D G$ is  uniformly continuous on $ \K^{\epsilon}_a$, i.e. there exists a
    function $\omega : \Rsp^+ \to \Rsp^+$ that satisfies $ \lim\limits_{x \to 0}
    \omega(x) = \omega(0) = 0$ and such that
    $$
     \forall \psi, \tilde\psi \in \K^{\epsilon}_a,~ \nr{DG(\psi) - DG(\tilde\psi)} \leq
    \omega(\nr{\psi - \tilde\psi}).
    $$
Note also that the modulus of continuity  $\omega$  can be assumed to be an increasing function.
    For any $\psi \in \K^{\epsilon}_a$, we let $ v = DG^{+}(\psi) (G(\psi) - \nu) $ and
     $ \psi_\tau = \psi - \tau v $ for any $\tau \geq 0 $.
   Since $G$ is of class $\Class^1$,
    a Taylor expansion in $\tau$ gives
    \begin{equation}
        \label{eq:taylor-g}
        G(\psi_\tau) = G(\psi - \tau DG^{+}(\psi) (G(\psi) - \nu)) = (1 - \tau)
        G(\psi) + \tau \nu + R(\tau)
    \end{equation}
    where $ R(\tau) = \int_0^\tau (DG(\psi_\sigma) - DG(\psi)) v d\sigma $ is
    the integral remainder. Then, we can bound the norm of $ R(\tau) $
    \begin{align*}
        \nr{R(\tau)} &= \nr{\int_0^\tau (DG(\psi_\sigma) - DG(\psi)) v d\sigma} \\
                     &\leq \nr{v} \int_0^\tau \omega(\nr{\psi_\sigma - \psi})
                     d\sigma = \nr{v} \int_0^\tau \omega(\sigma \nr{v}) d\sigma
                     \\
                     &\leq \tau \nr{v} \omega(\tau \nr{v})
    \end{align*}
    where we have used the fact that $ \omega $ is an increasing function.\\

    \noindent{\textbf{Step 1}}
    We first want to show that for every $\psi \in \K^{\epsilon}_a$ there exists $\tau(\psi)>0$ such that
    \begin{equation}\label{eq:cv-algo}
        \forall \tau \in (0, \tau(\psi))\quad \quad
    \psi_\tau \in \K^{\epsilon}_a
    \quad \mbox{and}\quad
    \nr{G(\psi_\tau) - \nu} \leq \left( 1 - \frac{\tau}{2} \right)
    \nr{G(\psi) - \nu}.
    \end{equation}
    Recall that for every $i \in \{1, \ldots, N\}$ one has $ \nu_i \geq 2 \epsilon
    $ and $ G_i(\psi) \geq \epsilon $. Thus one gets
    $$ G_i(\psi_\tau) \geq (1 - \tau) G_i(\psi) + \tau \nu_i + R_i(\tau) \geq (1
    + \tau) \epsilon - \nr{R(\tau)}.$$
    So if we choose $ \tau $ such that $ \nr{R(\tau)} \leq \tau \epsilon $ then
    $ G_i(\psi_\tau) \geq \epsilon $ and $ \psi_\tau \in \K^{\epsilon}$.
    Now, since $\lim\limits_{x \to 0} \omega(x) = 0 $, there exists $\alpha_1 > 0$ such that for every $ 0 \leq \sigma \leq \alpha_1$, one has
    $\omega(\sigma) \leq \epsilon / \nr{v} $.
    This implies that if $ \tau \leq \alpha_1
    / \nr{v} $, then $\nr{R(\tau)} \leq \tau \epsilon$ and consequently $ \psi_\tau \in \K^{\epsilon}$. Note that $G(\psi) - \nu$ belongs to $\cstperp$ and that $ DG(\psi) $ is an isomorphism from $\cstperp $ to $ \cstperp $. We deduce that $\psi_\tau - \psi = \tau v$ belongs to $\cstperp$, hence $\psi_\tau \in \K^{\epsilon}_a$.\\

     From Eq.~\eqref{eq:taylor-g}, we have
    $ G(\psi_\tau) - \nu = (1 - \tau) (G(\psi) - \nu) + R(\tau)$.
    So, to get the second condition of Equation \eqref{eq:cv-algo}, it is sufficient to show that
     $ \nr{R(\tau)} \leq
    (\tau/2) \nr{G(\psi) - \nu} $. The estimation on $ \nr{R(\tau)} $ and
    the definition of $ v $ gives us
    $$ \nr{R(\tau)} \leq \tau \nr{DG^{+}(\psi)} \nr{G(\psi) - \nu} \omega(\tau
    \nr{v}). $$
    Still from the continuity of $\omega$ at $0$, we can find $ \alpha_2 > 0 $ such
    that for every $\tau \leq \alpha_2 / \nr{v}$ one has  $\omega(\tau \nr{v}) \leq
    \epsilon / 2 \nr{DG^{+}(\psi)}$, thus
    $ \nr{R(\tau)} \leq (\tau/2) \nr{G(\psi) - \nu}$. Therefore, by putting $\tau(\psi):=\min(\alpha_1 /
    \nr{v(\psi)} ,\alpha_2 /
    \nr{v(\psi)} , 1) $, Equation \eqref{eq:cv-algo} is proved. Note that we impose $\tau(\psi)$ to be less than $1$.\\

    \noindent{\textbf{Step 2}}
    The function $ G $ is of class $ \Class^1 $ on $ \K^{\epsilon}_a$. For every $ \psi $ in $ \K^{\epsilon}_a$,
$ DG(\psi) $ is an isomorphism from $
    \cstperp $ to $ \cstperp $ and its inverse $ DG^{+}(\psi) $ depends
    continuously on  $ \psi $. Since  $\sum_i G_i(\psi) = \sum_i \nu_i $, $G(\psi) - \nu$ belongs to $\cstperp $,
    so the function $ v(\psi) = DG^{+}(\psi) (G(\psi) -
    \nu) $ is also continuous by composition.  If $G(\psi) \neq \nu$, the strict
    monotonicity of $ G $ ensures that $ v(\psi) \neq 0 $ and so $\tau(\psi)=\min(\alpha_1 /
    \nr{v(\psi)} ,\alpha_2 / \nr{v(\psi)} , 1) $ is also continuous in $\psi$. If $G(\psi)=\nu$, then $v(\psi)=0$. However, by continuity of $v$, the function $\tilde{\psi}\mapsto \tau(\tilde{\psi})$ is constant equal to $1$ in a neighborhood of $\psi$. Hence the function $\psi \mapsto \tau(\psi)$ is globally continuous.
     Therefore, the infimum of $\tau(\psi)$ over the compact set $ \K^{\epsilon}_a
     $ is attained at a point of $ \K^{\epsilon}_a$, thus is strictly positive.
     We deduce that we can take a uniform bound $\tau(\psi)=:\tau^*>0 $ in  Equation~\eqref{eq:cv-algo} that does not depend on $\psi$.
    This directly implies the convergence of the damped Newton algorithm with linear speed.
\end{proof}


\begin{proof}[Proof of Theorem~\ref{coro:convergence-newton}]
The function $G$ appearing in \eqref{eq:MA} satisfies the regularity
condition (Theorem \ref{thm:regularity}) and the monotonicity condition
(Theorem \ref{thm:strict-concavity}) needed in
Proposition~\ref{prop:convergence-newton}. It remains to show the compactness
condition.
Let us take $ a \in \Rsp $ and let us show that $ \K_a^{\epsilon} $
is compact. It is easy to see that $ \K_a^{\epsilon} $ is closed since $ G $ is
continuous. Let $\psi\in \K_a^{\epsilon}$, $ i \neq j $ and $x\in \Lag_i(\psi)$. Then one has
$$
\psi_i \leq \psi_j + \nr{x-y_j}^2 - \nr{x-y_i}^2 \leq \psi_j + \diam(K\cup Y)^2,
$$
where  $\diam(K\cup Y)$ is the diameter of $K\cup Y$. 
So the differences $ | \psi_i - \psi_j | $ are bounded by $\diam(K\cup Y)^2$. Combined with the fact that $ \sum_i \psi_i $ is constant, one has that $ \psi $ is bounded by a constant independent on $\psi$. Thus, $ K_a^{\epsilon} $ is compact.
\end{proof}


%% file: inc/4-numerical-results.tex
\section{Numerical results}
\label{sec:numerical-results}

In this section, we solve the optimal transport problem in $\Rsp^3$ between triangulated surfaces (possibly with holes,
with or without a  boundary) and point clouds, for the quadratic cost
\jocelyn{
and show it can be used in different settings: \emph{optimal quantization} of a
probability density over a surface, \emph{remeshing} and \emph{point set registration} on a mesh.
}
The source density is assumed to be affine on each triangle of the triangulated surface.
One crucial aspect of the algorithm is the exact computation of the combinatorics of the Laguerre cells, \textit{i.e.} the intersection between a triangulated surface and a 3D power diagram, see Equation~\eqref{eq:lag}. Another important aspect is the initialization step in  Algorithm~\ref{algo:newton}, \textit{i.e. }finding a set of weights $ \psi^0 $ which guarantees that all the initial Laguerre cells have a positive mass.
We first explain the algorithm to compute the Laguerre cells, describe how we take the initial weights, before presenting some results.

\subsection{Implementation}


We describe here briefly an algorithm to compute the combinatorics of the
intersection of a Power diagram $\Pow := (\Pow_i)_i$ with a triangulated surface
$K=\cup_{\sigma \in \Sigma} \sigma$, with triangles $ \sigma \in \Sigma $.
Note that in general the intersection of a power cell with $K$ is not convex and
can even have several connected components (as illustrated for instance in
Figure \ref{fig:experiments-uniform}, in the second and third rows). We encode here the
triangulated surface $K$ with a connected graph $ G_1 $ where $ G_1 $ is the
$1$-skeleton of $ K $ (seen as a subset of $ \Rsp^3 $). Similarly, the intersection of the 2D
faces of the power diagram with the triangulated surface $K$, namely $G_2=\cup_i ( K
\cap \partial \Pow_i )$, is also encoded by a graph. Let us remark
that $ G_2 $ can be disconnected.
More precisely, one proceeds as follows:
\begin{enumerate}
    \item We first split the edges in the graph $ G_1 $ at points in $ G_1\cap
        G_2 $.  Since $ G_1 $ is connected, this can be done by a simple traversal, in which
        we need to intersect the edges of the triangulation with the 2-dimensional
        power cells.
    \item We then traverse $ G_2 $ starting from vertices in $ G_1\cap G_2 $ by
        intersecting the 2-dimensional power-cells with triangles. $ G_2 $ might be
        disconnected, but we can discover the connected components using the
        non-visited vertices in $ G_1\cap G_2 $. This step provides us with both the
        geometry and connectivity of $ G_1\cup G_2 $, and also an orientation coming
        from the underlying triangulated surface $ K $.
    \item The graph $ G_1 \cup G_2 $ is embedded on the triangulated surface $ K
        $, and the connected components of $K\setminus (G_1\cup G_2) $ are (open) convex
        polygons. Each of these polygons represents an intersection of the form
        $ \Pow_i \cap \sigma $. The boundary of these polygons can easily be
        reconstructed from $ G_1\cup G_2 $ and the orientation (obtained in the
        second step).
\end{enumerate}

The main predicates needed here are the intersection tests between a
2D face and a segment and between a power edge (1D face) and a triangle. These
predicates can easily be implemented in an exact manner using, for example,
the filtered predicates mechanism provided by the \texttt{CGAL} library \cite{cgal}.

\vspace{.3cm}\noindent \textbf{Numerical integration.}
The computation of $ G_i(\psi) $ requires the evaluation of integrals of the
form $ \int_{\Lag_i(\psi) \cap \sigma} \rho_\sigma(x) \dd \Haus^2(x) $ where $
\rho_\sigma : \Rsp^3 \to \Rsp^+ $ is an \jocelyn{affine} density. In order to evaluate these integrals exactly, we use the
classical Gaussian quadrature formulae. In our setting, we have that if $ t =
[a, b, c] $ is a triangle and $\rho:t\to\Rsp$ is affine, then $ \int_{t} \rho(x)
\dd \Haus^2(x) = \Area(t) \cdot \rho((a +b + c)/3) $.

\vspace{.3cm}\noindent \textbf{Choice of the initial weights.}  The
following proposition shows how to choose the initial weights so
as to avoid empty Laguerre cells.  

\begin{proposition} Let $ K \subset \Rsp^d $ be a
compact set, $ Y = \{y_1, \ldots, y_N \} \subset \Rsp^d $ be a point
cloud and $ \psi^0_i = -\dd(y_i,K)^2 $. Then, all the Laguerre cells
$\Lag_i(\psi^0)$ are non-empty:
$$ \emptyset \neq \{x\in
K \mid \dd(y_i,K)=\nr{x-y_i}\} \subseteq \Lag_i(\psi^0). $$
    \end{proposition}

\begin{proof}
   Let $ i \in \{1, \ldots, N \} $, and $ x \in K
    $ be such that $ d(y_i, K) = \nr{x - y_i}
    $.
    Then for $ j \in \{1, \ldots, N \} $
    $$
        \nr{x - y_j}^2 + \psi^0_j 
        = \nr{x - y_j}^2 -
        d(y_j, K)^2 
        \geq d(y_j, K)^2 - d(y_j, K)^2 
        = 0 
        = \nr{x - y_i}^2 +
        \psi^0_i,
        $$
meaning that $ x \in \Lag_i(\psi^0) $.\end{proof}

In particular, this proposition applies to the case where $ K $ is a
triangulated surface. Thus, it means that we can find weights such that the initial
Laguerre cells are not empty. In practice, if needed, we slightly perturb $ \psi^0 $ to
ensure that all the Laguerre cells also have non empty interior, thus have a positive mass.

\subsection{Results and applications}
We compute the optimal transport map between a piecewise linear measure defined on a triangulated surface $ K $ and a
discrete measure defined on a 3D point cloud. Even if we can handle non uniform
measures, in the examples presented here, the source density is uniform over the
triangulation: $ \rho_\sigma = 1 / \Area(K) $ for every $ \sigma \in \Sigma$,
where $\Area(K)$ is the area of $K$. The point cloud is chosen to be a noisy
version of points sampled on the mesh. In the examples, \jocelyn{the solutions are
computed up to an error of $\eta = 10^{-6} $.}

The first two rows of Figure \ref{fig:experiments-uniform} displays
results for a uniform target measure and the last two for a
non-uniform one.  Remark that in this case the non uniformity creates
smaller Laguerre cells on the right side. Note that the centroids of
the Laguerre cells provide naturally a correspondence between the
point cloud and the triangulated surface: we associate to each $y_i$
the centroid of the Laguerre $\Lag_i(\psi^k)$, where $\psi^k$ is the
output of Algorithm~\ref{algo:newton}.  In practice, the number of
iterations remains small even for large point sets. For instance, if
we choose $ 10,000 $ noisy samples on the torus, the algorithm takes $
16 $ iterations to solve the problem.

\begin{remark}
We also underline that the Laguerre cells can be non \jocelyn{geodesically} convex and even
disconnected (as illustrated in the second and third \jocelyn{columns} of Figure
            \ref{fig:experiments-uniform}) which shows that our method
            handles more general settings than \cite{kitagawa2016newton},
            \textit{i.e.}  cost functions whose Laguerre cells cannot be convex in any chart
            (violating the hypothesis of \cite{kitagawa2016newton}, Def 1.1).
\end{remark}

\begin{figure}[h]
    \centering
    \includegraphics[width=.24\textwidth]{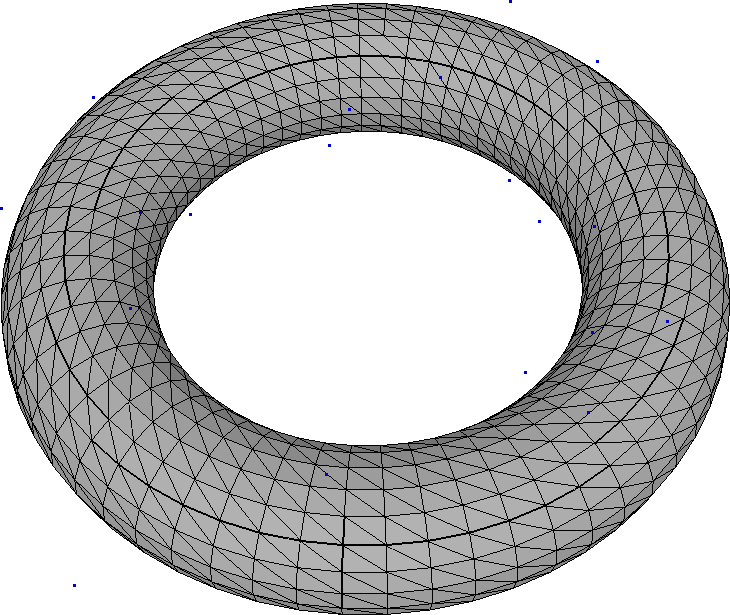}
    \includegraphics[width=.24\textwidth]{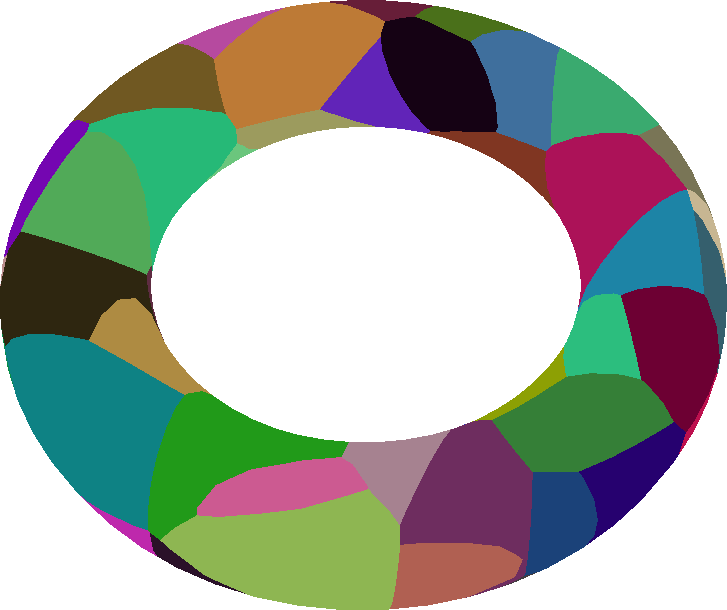}
    \includegraphics[width=.24\textwidth]{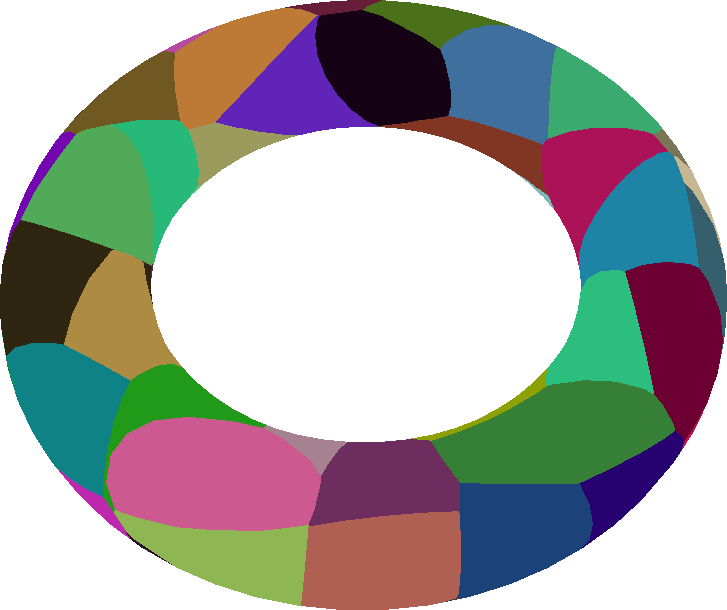}
    \includegraphics[width=.24\textwidth]{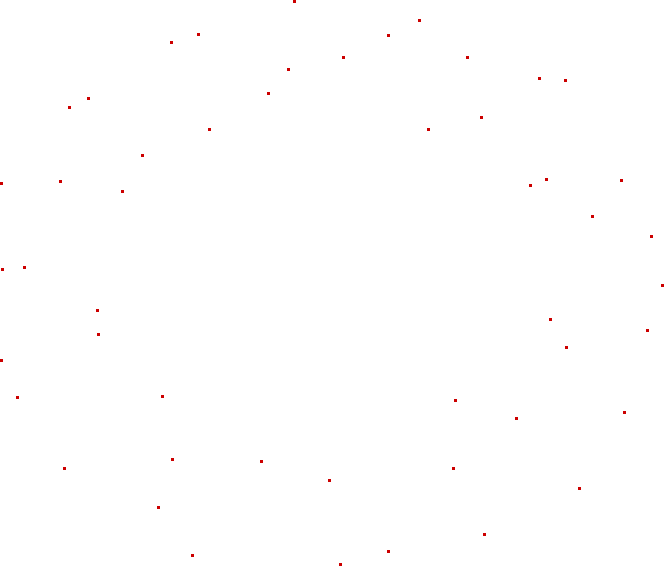}

    \includegraphics[width=.24\textwidth]{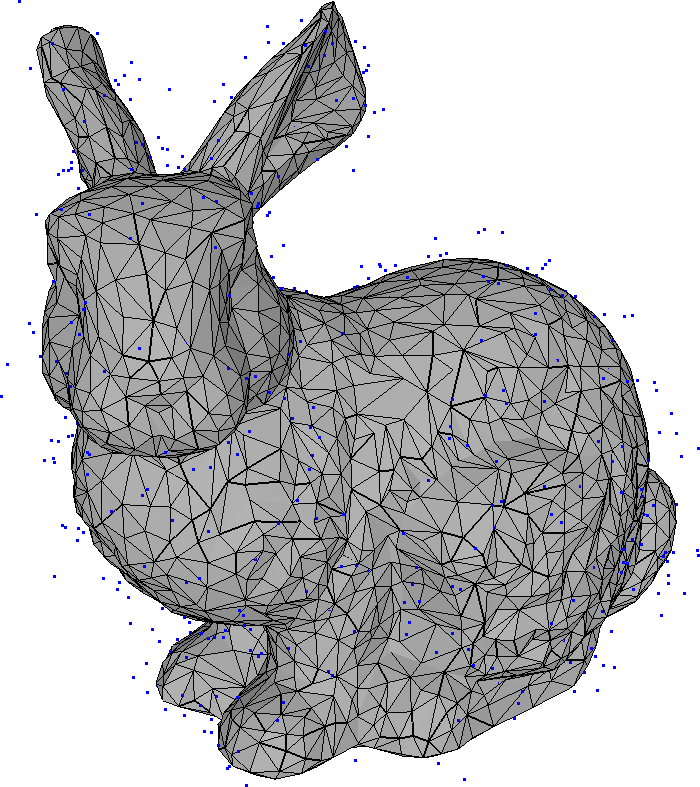}
    \includegraphics[width=.24\textwidth]{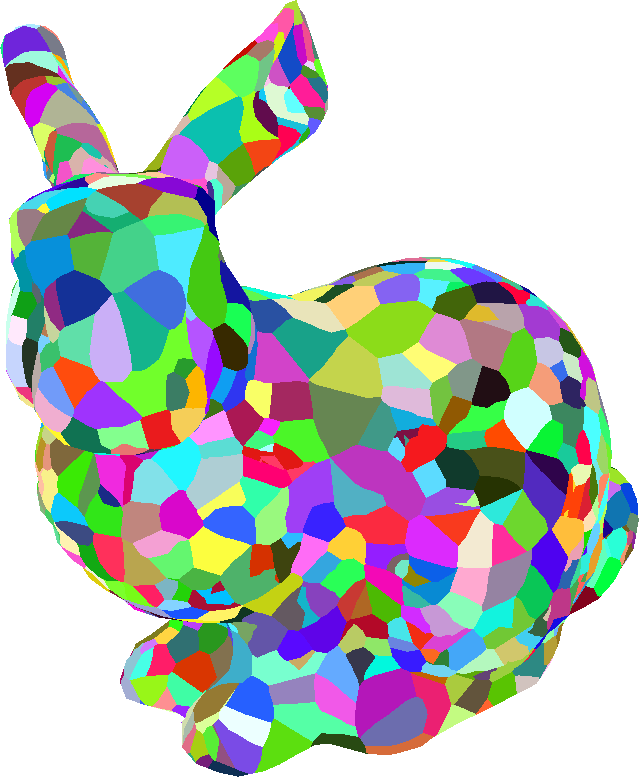}
    \includegraphics[width=.24\textwidth]{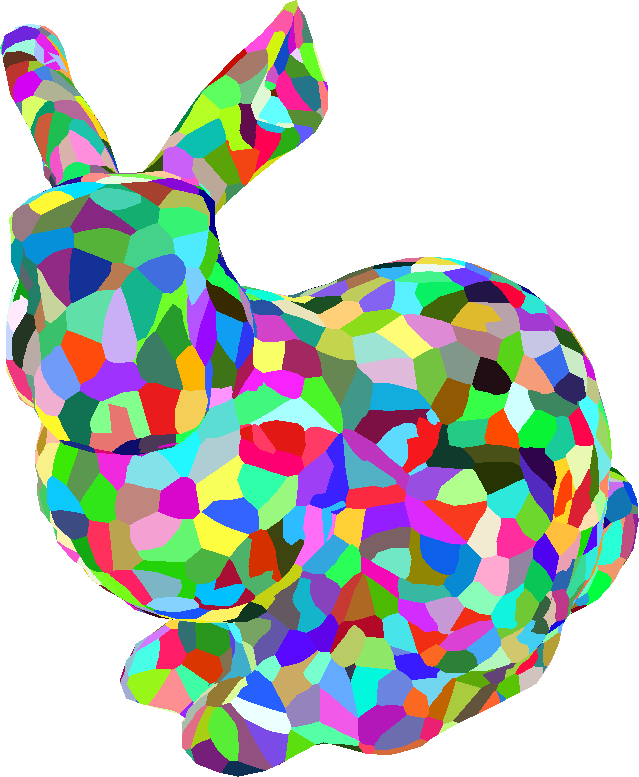}
    \includegraphics[width=.24\textwidth]{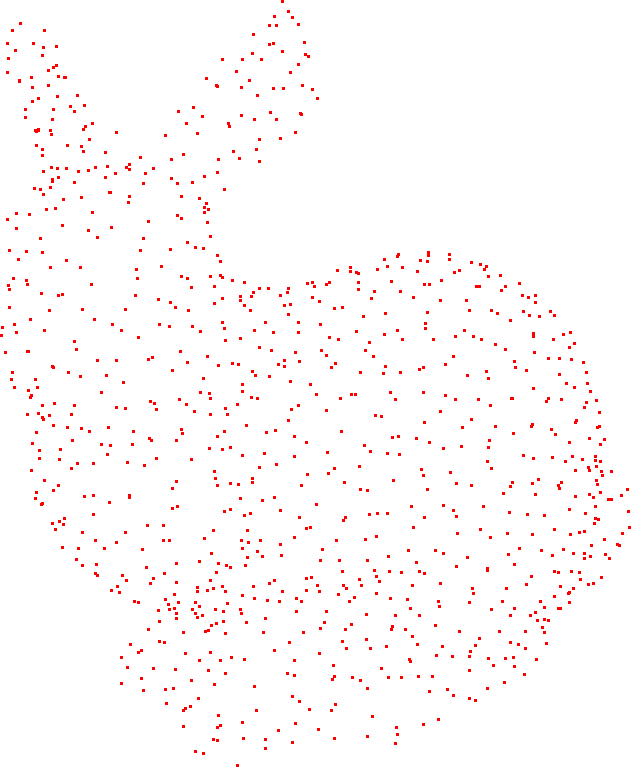}

    \includegraphics[width=.24\textwidth]{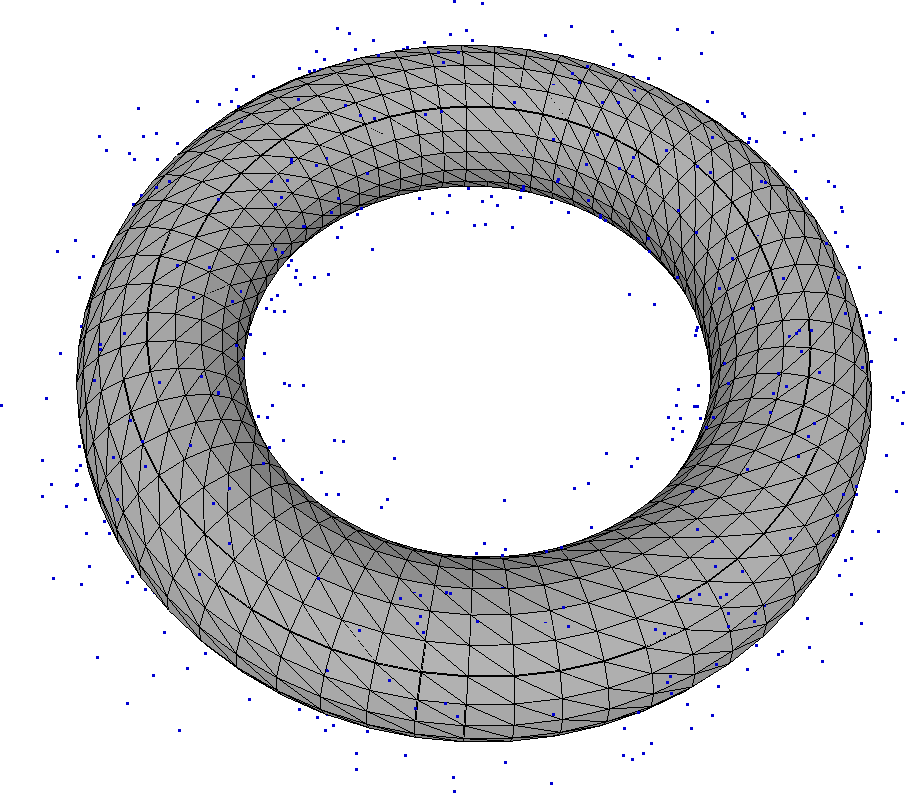}
    \includegraphics[width=.24\textwidth]{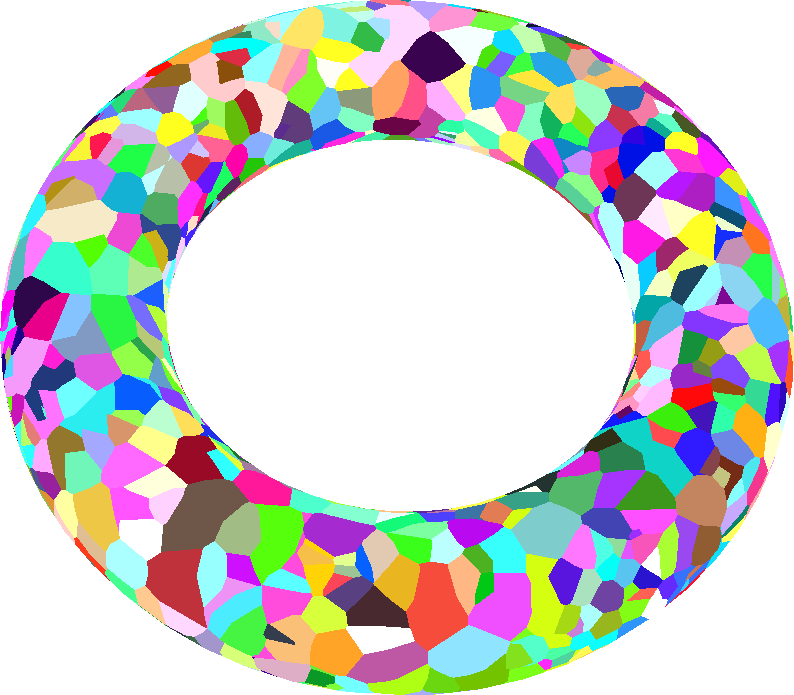}
    \includegraphics[width=.24\textwidth]{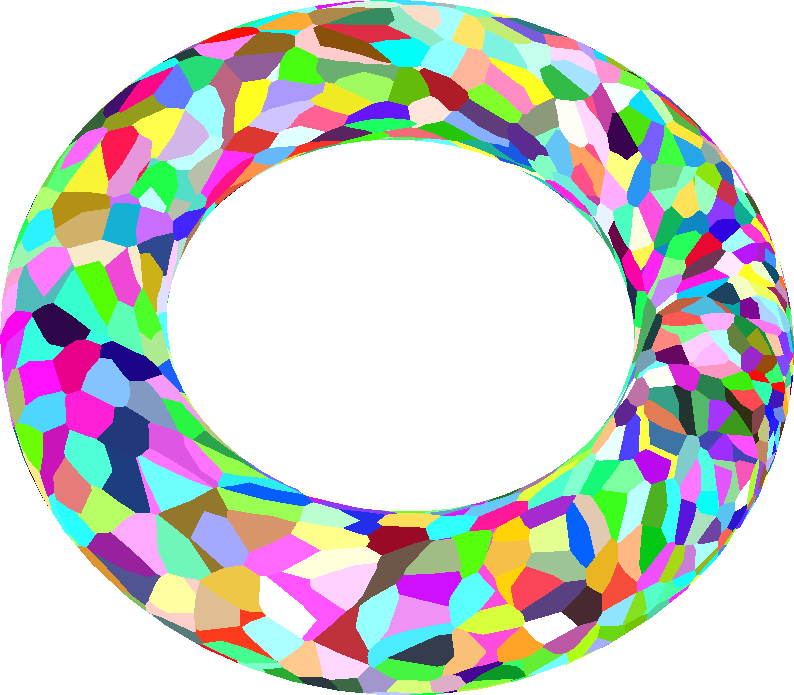}
    \includegraphics[width=.24\textwidth]{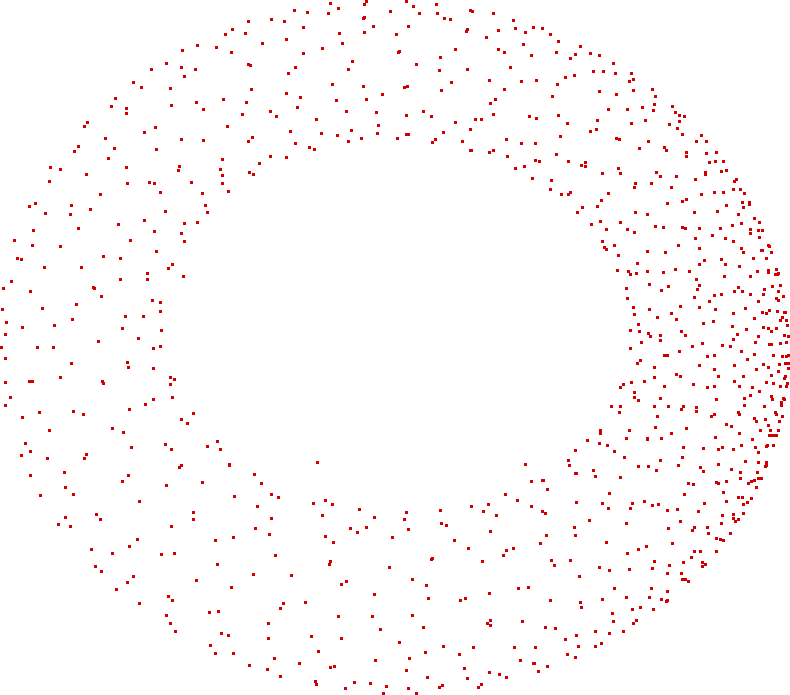}

    \includegraphics[width=.24\textwidth]{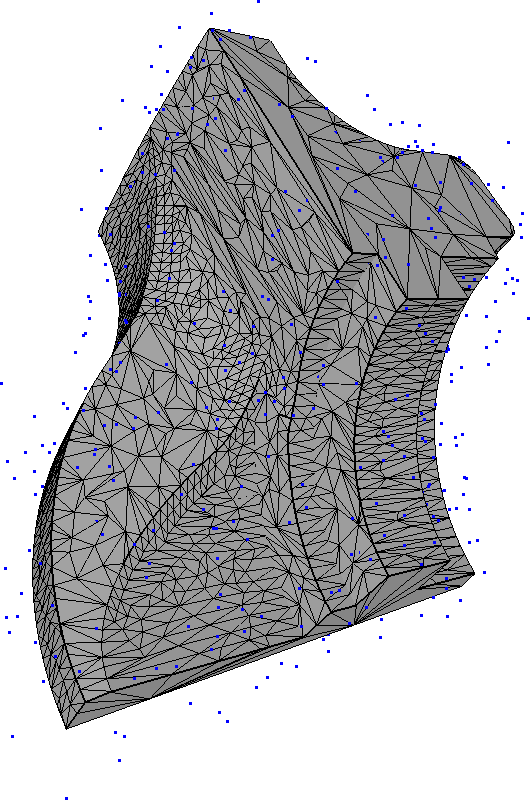}
    \includegraphics[width=.24\textwidth]{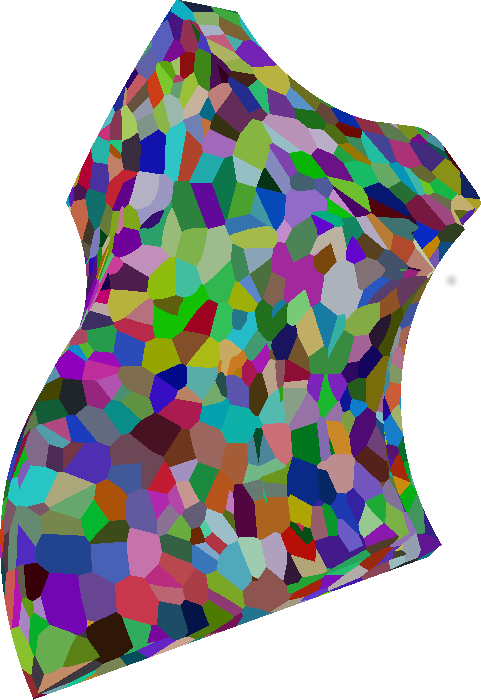}
    \includegraphics[width=.24\textwidth]{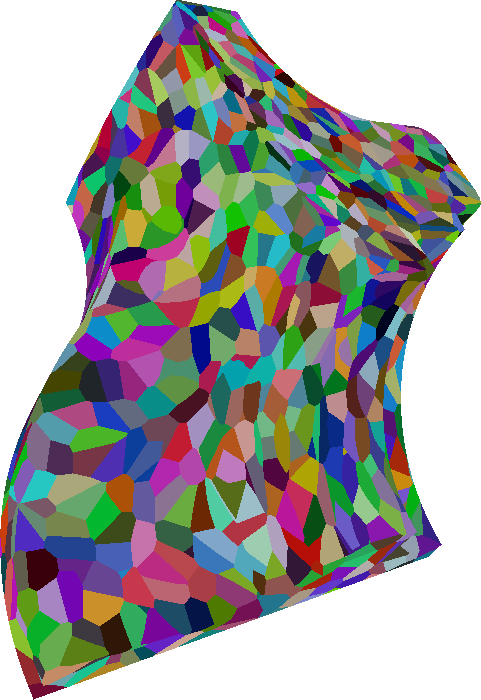}
    \includegraphics[width=.24\textwidth]{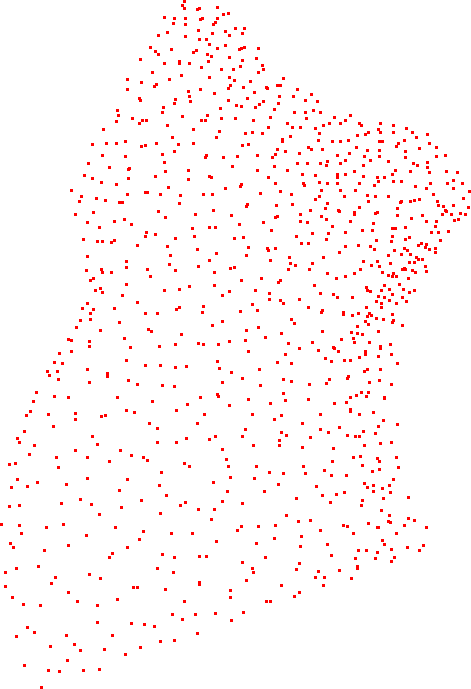}

    \caption{From left to right: Mesh and initial point cloud (in blue), Initial
    Laguerre cells, Final Laguerre cells, Centroids of the final Laguerre
    cells.
    The source measure is uniform. In the first two rows, the target
    \jocelyn{density} is uniform while in the last two, it linearly decreases from left to right.
    In the first row, $ N = 50 $ while in the other rows, $ N = 1000 $.
    Computation time (number of iterations): 3s (4) /  41s (7) / 74s (14) / 58s
    (9).}
    \label{fig:experiments-uniform}
\end{figure}

\jocelyn{ We now show how to use this algorithm as a building block
  for higher level operations: optimal quantization of surfaces,
  remeshing and point set registration. The goal here is not to
  compete with state of the art algorithms for these applications but
  rather to show the applicability of Algorithm \ref{algo:newton} in
  more complex situations.

\subsubsection{Optimal quantization of a surface}

\emph{Optimal quantization} is a sampling technique used to
approximate a density function with a point cloud, or more accurately
a finitely supported measure. It has many applications such as in
image dithering or in computer graphics (see \cite{de2012blue} for
more details).  Here, we show how to perform this kind of sampling on
triangulated surfaces.  Given a triangulated surface $ K \subset
\Rsp^3 $ and a density $ \mu $ on $K$, we first define $ Y_0 $ as the
set of vertices of $ K $ and consider the constant probability measure
$ \nu_0 $ on $ Y_0 $. For each $ k \ge 0 $, we solve the optimal
transport between $ \mu $ on $ K $ and $ \nu_k $ on $ Y_k $ and pick
one point, for instance the centroid, per Laguerre cell.  We iterate
this procedure by choosing for the new point cloud $ Y_{k+1} $ the set
of the previously computed centroids and for $\nu_{k+1} $ the uniform
measure over $ Y_{k+1} $.  After a few iterations, this gives us a
(locally) optimal quantization of $ K $.  Figure
\ref{fig:applications-blue-noise} shows examples of sampling on
different surfaces with different densities.

\begin{figure}[h]
    \centering
    \includegraphics[height=.2\textwidth]{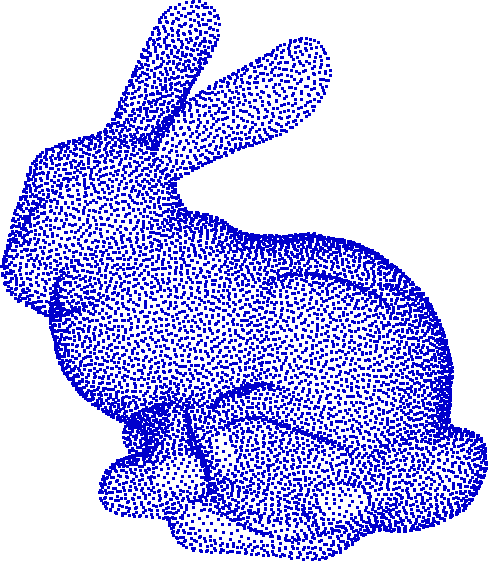}
    \includegraphics[height=.2\textwidth]{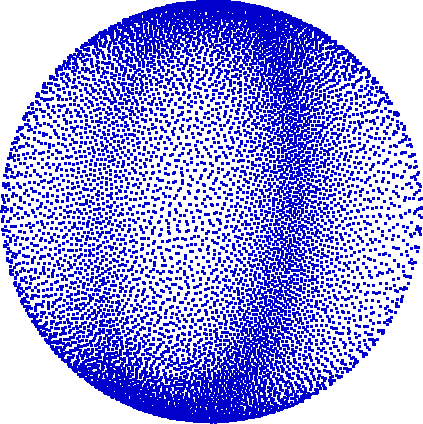}
    \includegraphics[height=.2\textwidth]{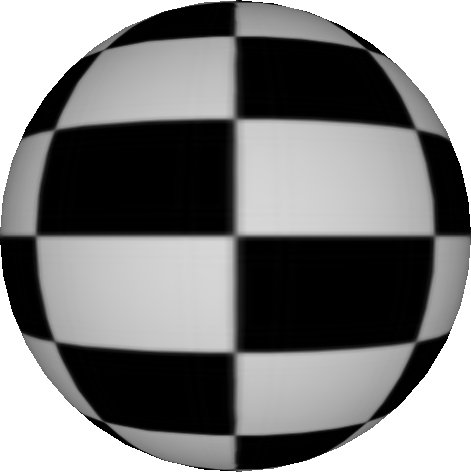}
    \includegraphics[height=.2\textwidth]{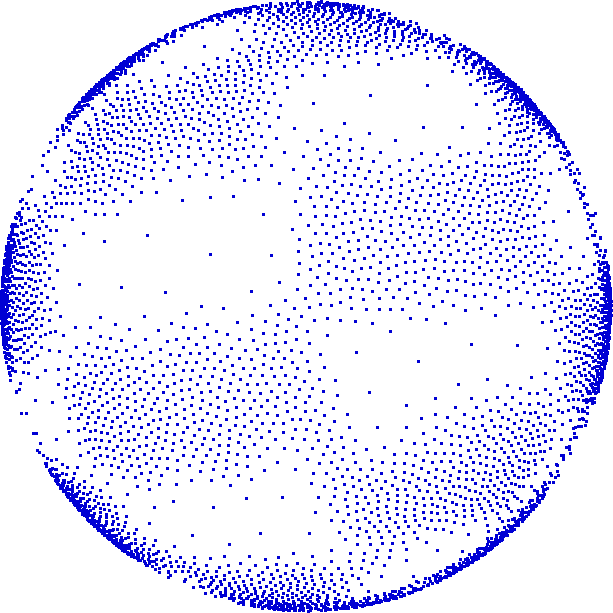}
    \caption{Optimal quantization of triangulated surfaces for different
        densities and surfaces. From left to right: uniform density $ \mu = 1 $ on the
        Stanford Bunny (10k points); non-linear density $\mu(x, y,z) = e^{-3|y|} $ on the sphere (10k
        points); checkerboard texture and sampling for the density corresponding
        to the UV-mapping of the texture on the hemisphere (5k points).}
    \label{fig:applications-blue-noise}
\end{figure}

\subsubsection{Remeshing}


We now consider the following problem: given a triangulated surface $ K $, a density
$ \mu $ supported on this mesh, we want to build a new mesh such that the
distribution of triangles respect this density, meaning that we want more
triangles where the density is more important. This has applications for
instance in finite element methods for solving partial differential equations
where the quality of the discretization matters. To do this, we can use the
following simple procedure: we consider the uniform discrete measure $\nu$ supported on the vertices of $K$; we solve
the optimal transport between $ \mu $ on $K$ and $ \nu $; the new mesh will be taken as the dual (in the graph sense) of the final Laguerre diagram.
See Figure \ref{fig:applications-remeshing} for two examples for different source
densities.

\begin{figure}[h]
    \centering

    \includegraphics[width=.23\textwidth]{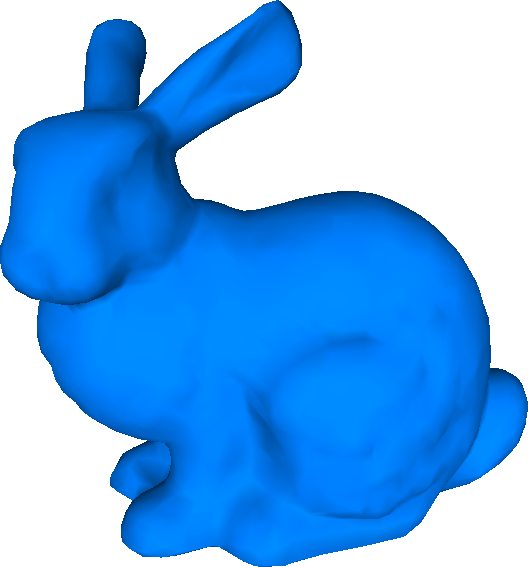}
    \includegraphics[width=.5\textwidth]{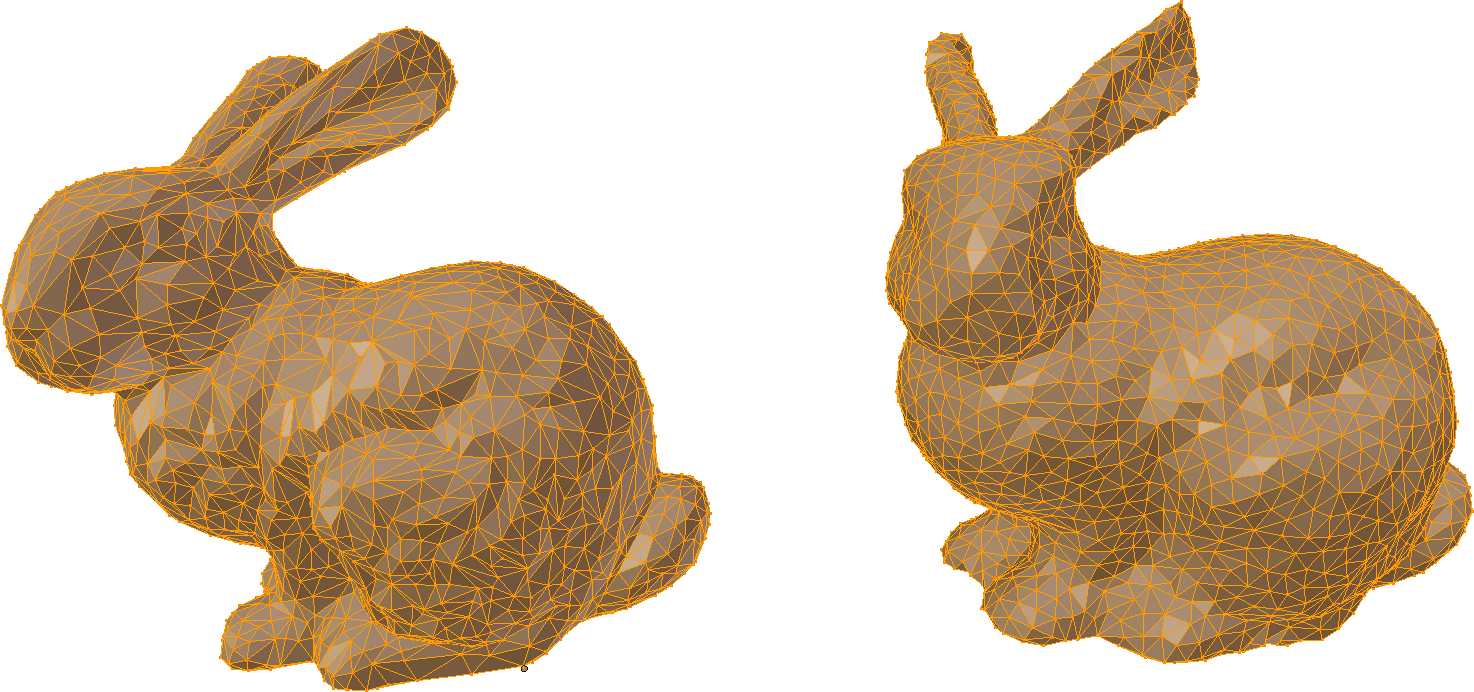}

    \includegraphics[width=.23\textwidth]{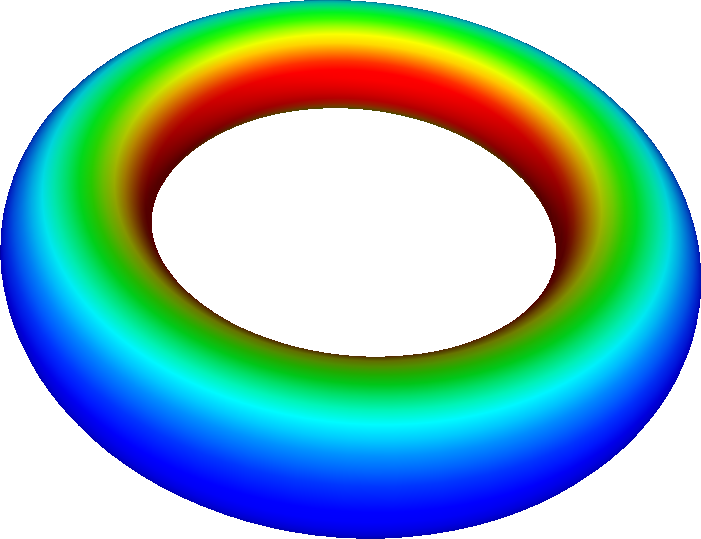}
    \includegraphics[width=.5\textwidth]{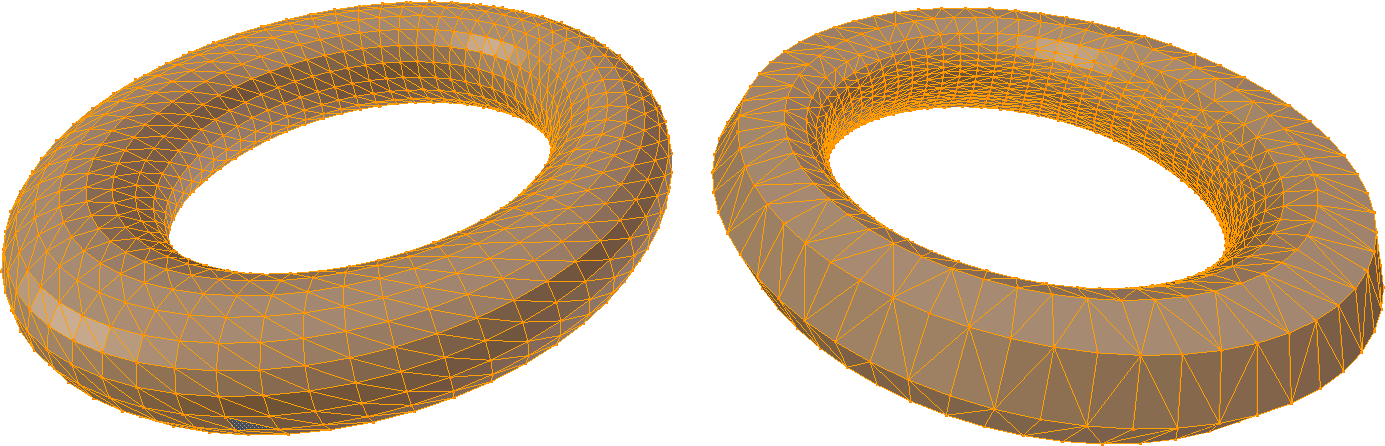}

    \caption{Remeshing using optimal transport. From left to right: source
        density; initial mesh and remeshed surface. First row: Uniform density: $ \mu = 1 $; Second
        row: $ \mu $ is proportional to a mean curvature estimator
        of the source mesh. Number of vertices for each model: Bunny: 2.2k;
        Torus: 5.6k.}
    \label{fig:applications-remeshing}
\end{figure}

\subsubsection{Point set registration}

We finally consider the rigid point set registration to a mesh. Given
a triangulated surface $ K $ and a point cloud $ Y $, we want to find
a rigid transformation $ T $ such that the $L^2$ distance between $ K
$ and $ T(Y) $ is minimal. The most popular method to do this is the
Iterative Closest Point (ICP) algorithm developed in
\cite{besl1992method}. For this algorithm, we need to be able to
compute for each point $ y_i $ from the cloud $ Y $ its closest point
on the mesh $ K $. We can replace the traditional nearest neighbor
query with the following routine: we solve the optimal transport
between the constant probability measure $ \mu $ on $ K $ and the
constant probability measure $ \nu $ on $ Y $, then associate each
point $ y_i $ to a point (for instance the centroid) of the Laguerre
cell $ \Lag_i(\psi) $ where $ \psi \in \Rsp^N $ are the final
weights. The resulting algorithm is called Optimal Transport ICP
(OT-ICP).  See Figure \ref{fig:applications-icp} for one example. In
our results, OT-ICP converges in much less iterations than standard
ICP, namely 3 iterations versus 20 iterations for the same stopping
criterion in our two test cases.  Despite this, the remains higher
with optimal transport.  One may hope that the use of optimal
transport "convexifies" the energy optimized by ICP, in the same way
the choice of a Wasserstein loss instead of a $\mathrm{L}^2$ distance
seems to mitigates the cycle-skipping issue in full waveform inversion
\cite{engquist2014application}.

\begin{figure}[h]
    \centering
    \includegraphics[height=.15\textheight]{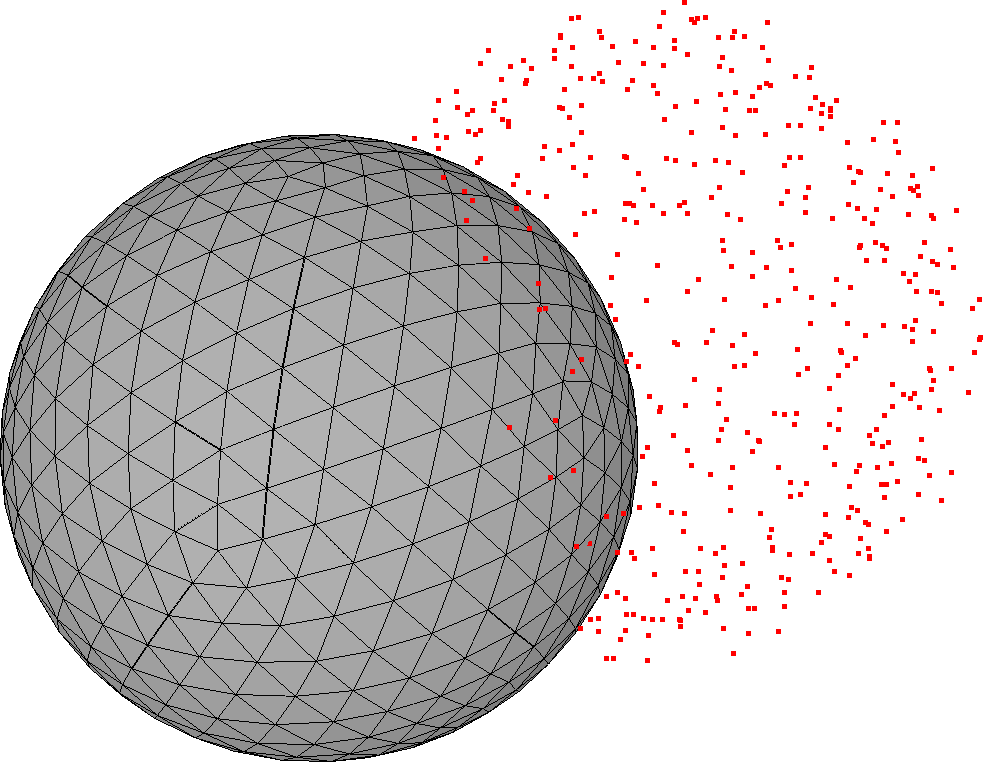}
    \includegraphics[height=.15\textheight]{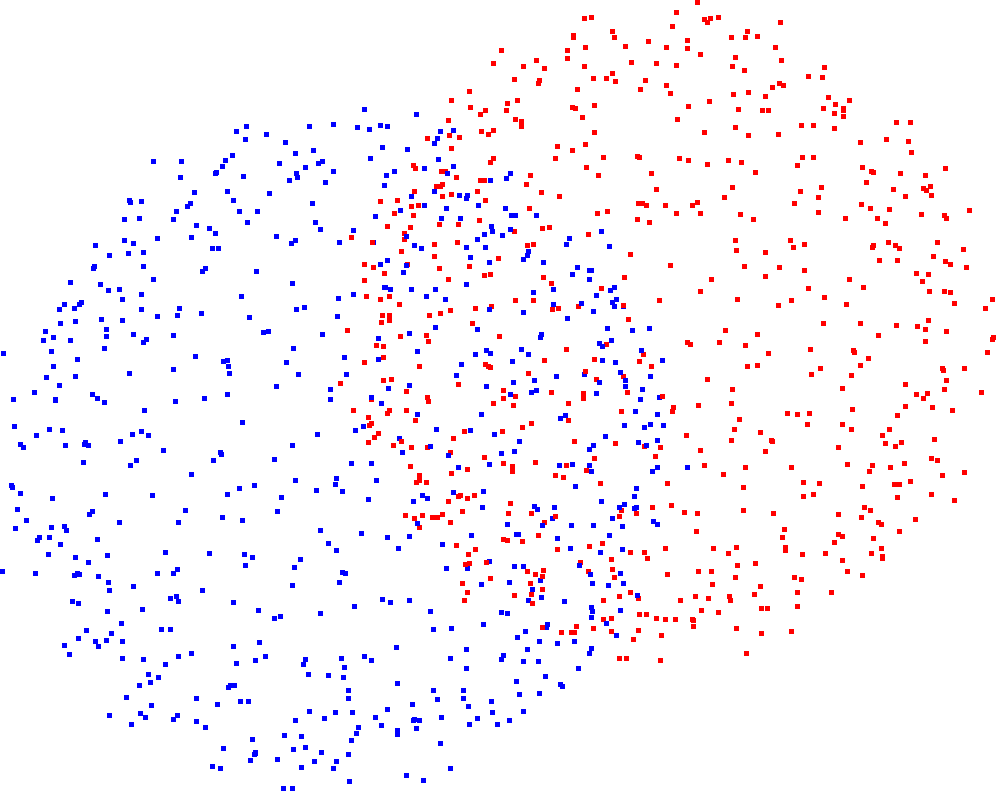}
    \includegraphics[height=.15\textheight]{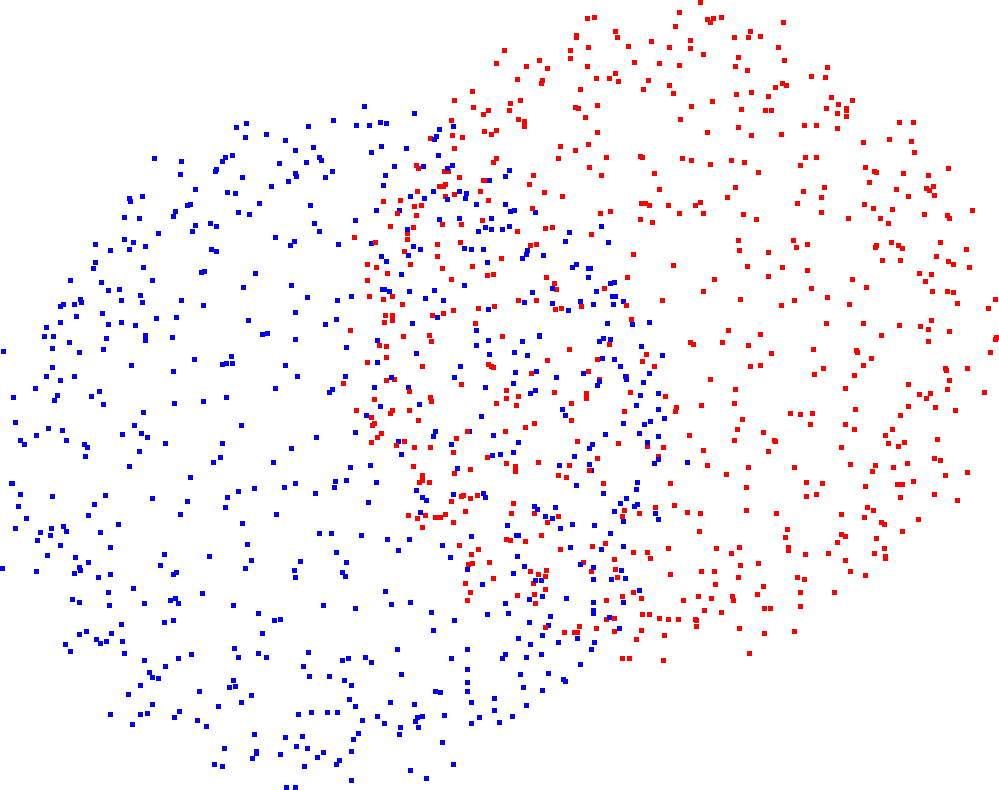}

    \includegraphics[height=.15\textheight]{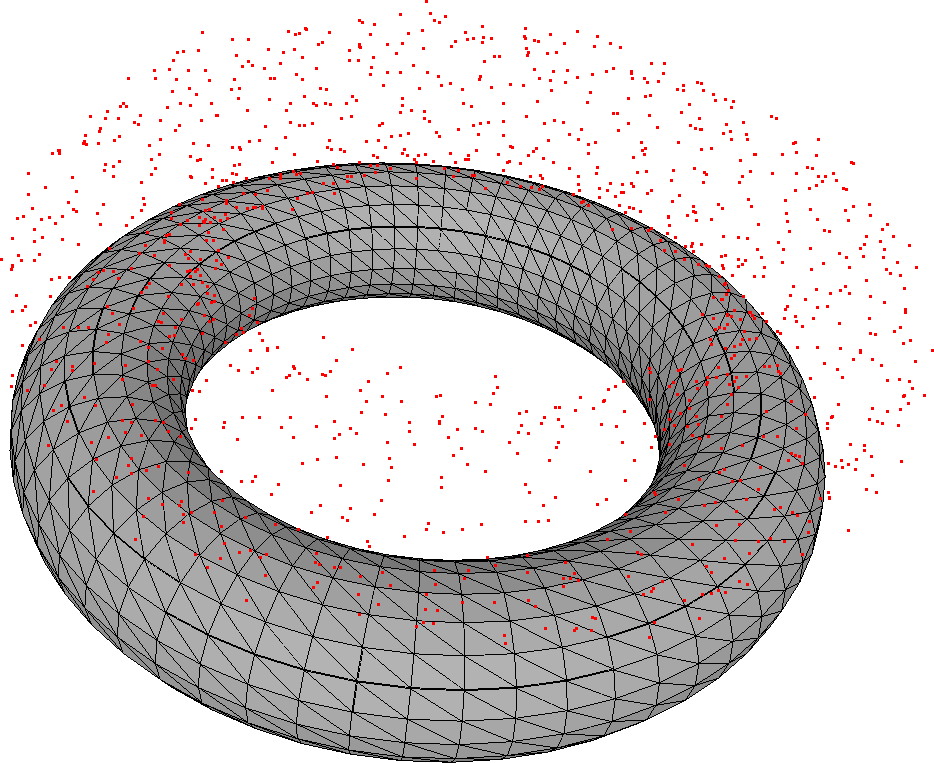}
    \includegraphics[height=.15\textheight]{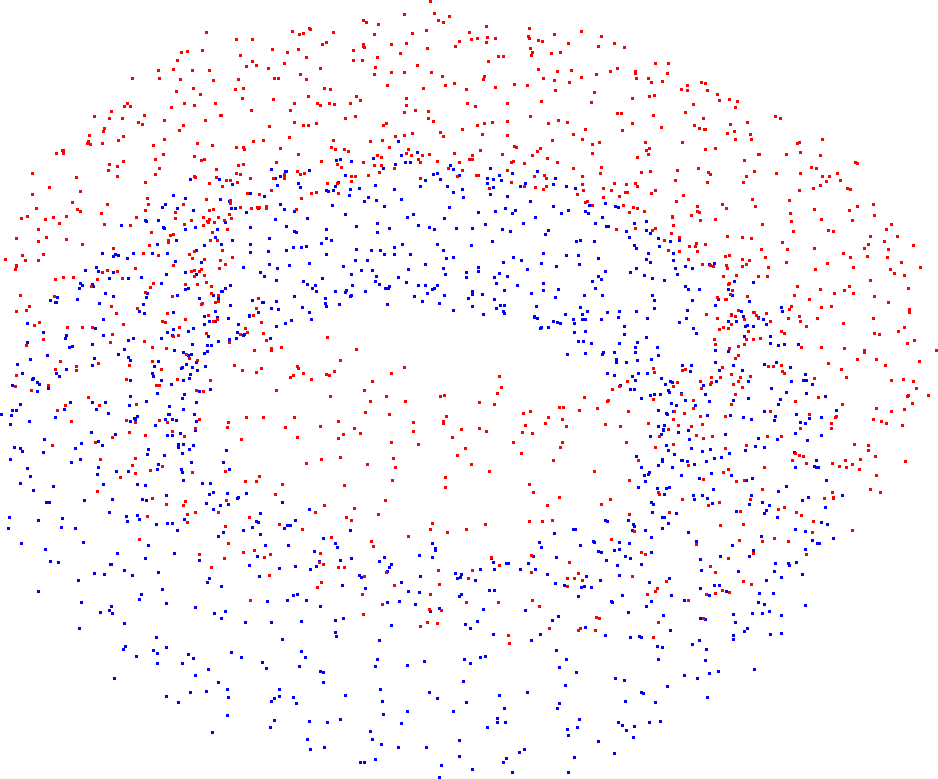}
    \includegraphics[height=.15\textheight]{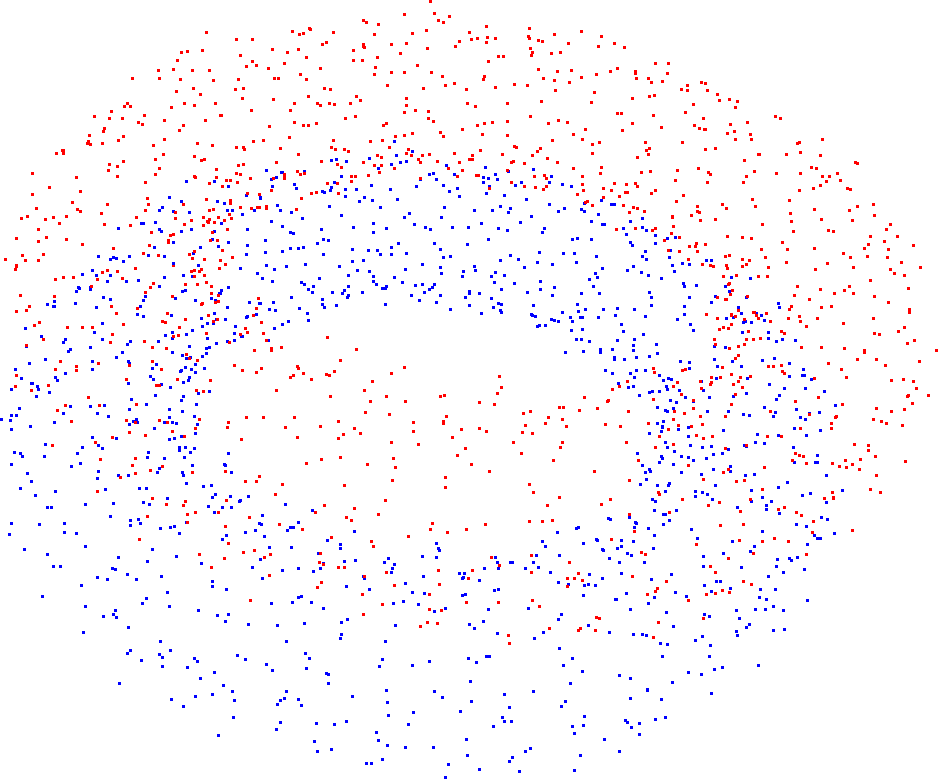}

    \caption{Optimal transport-ICP algorithm. From left to right:
      initial mesh (in grey) and initial point cloud (in red); initial
      (red) and final (blue) point clouds using traditional ICP;
      initial (red) and final (blue) using optimal transport.}
    \label{fig:applications-icp}
\end{figure}
}
